\documentclass[12pt]{article}

\usepackage[inline,shortlabels]{enumitem}
\usepackage{float}
\usepackage{mathtools}
\usepackage{subdepth}
\usepackage{graphicx}
\usepackage[round]{natbib}
\usepackage{bibunits}
\usepackage[round]{natbib}
\usepackage[margin=1in]{geometry}
\usepackage{tikz}
\usetikzlibrary{arrows.meta}
\usetikzlibrary{decorations.markings}
\usepackage[linesnumbered]{algorithm2e}
\RestyleAlgo{ruled}
\usetikzlibrary{plotmarks}
\usetikzlibrary{shapes,decorations,arrows,calc,arrows.meta,fit,positioning}
\usetikzlibrary{shapes.geometric}
\tikzset{
font={\fontsize{7pt}{12}\selectfont},
every node/.style={scale=1.3},
square/.style={regular polygon,regular polygon sides=4},
    -Latex,auto,node distance =1 cm and 1 cm,semithick,
    state/.style ={ellipse, draw, minimum width = 0.7 cm},
    block/.style = {square, draw, inner sep=0cm,minimum size=8mm},
    free/.style = {circle, draw, inner sep=0cm,minimum size=6mm},
    bidirected/.style={Latex-Latex,dashed},
    el/.style = {inner sep=2pt, align=left, sloped}
}\usepackage[english]{babel}
\usepackage{longtable}
\usepackage{nccmath}
\usepackage{color}
\usepackage{mathtools}
\usepackage{amssymb,amsmath,amsthm}
\usepackage{multirow}
\usepackage[titletoc,title]{appendix}
\usepackage{authblk}
\usepackage{bbm}
\usepackage{setspace}
\usepackage{dsfont}
\usepackage[OT1]{fontenc}
\usepackage{subcaption}
\usepackage{refcount}
\usepackage{booktabs}
\usepackage{array}
\graphicspath{ {../simulations/01_init_manuscript/graphs/} }
\usepackage{xr}
\usepackage[colorlinks,citecolor=blue,urlcolor=blue]{hyperref}

\newtheorem{remark}{Remark}
\newtheorem{proposition}{Proposition}
\newtheorem{theorem}{Theorem}
\AtEndDocument{\refstepcounter{theorem}\label{finalthm}}
\AtEndDocument{\refstepcounter{equation}\label{finaleq}}

{
  \theoremstyle{definition}
  
}
{
  \theoremstyle{definition}
  \newtheorem{assumptioniden}{}
}
{
  \theoremstyle{definition}
  
}

\newtheorem{lemma}{Lemma}
\AtEndDocument{\refstepcounter{lemma}\label{finallemma}}
\newtheorem{coro}{Corollary}
\newtheorem{definition}{Definition}

 \DeclareMathOperator{\var}{\mathsf{Var}}

\renewcommand{\P}{\mathsf{P}}
\newcommand{\p}{t}

\newcommand{\F}{\mathsf{F}}

\newcommand{\U}{\mathsf{U}}

\newcommand{\indep}{\mbox{$\perp\!\!\!\perp$}} 
 \newcommand{\dd}{\,\mathrm{d}}

\newcommand{\D}{\mathsf{D}}

\newcommand{\one}{\mathds{1}}
 
\newcommand{\E}{\mathsf{E}}
\renewcommand{\P}{\mathsf{P}}

\usepackage{accents}

\renewenvironment{proof}{{\it Proof }}{\qed \\}

\DeclarePairedDelimiterX{\norm}[1]{\lVert}{\rVert}{#1}

\pgfdeclarelayer{background}
\pgfsetlayers{background,main}
\usetikzlibrary{arrows,positioning}
\tikzset{
>=stealth',
punkt/.style={
rectangle,
rounded corners,
draw=black, very thick,
text width=6.5em,
minimum height=2em,
text centered},
pil/.style={
->,
thick,
shorten <=2pt,
shorten >=2pt,}
}
\newcommand{\Vertex}[3]
{\node[minimum width=0.6cm,inner sep=0.05cm] (#2) at (#1) {#3};
}
\newcommand{\Vertexr}[3]
{\node[rectangle, draw, minimum width=0.6cm,inner sep=0.05cm] (#2) at (#1) {#2};
}
\newcommand{\ArrowR}[3]%
{ \begin{pgfonlayer}{background}
\draw[->,#3] (#1) to[bend right=30] (#2);
\end{pgfonlayer}
}
\newcommand{\ArrowLW}[3]%
{ \begin{pgfonlayer}{background}
\draw[->,#3] (#1) to[bend left=30] (#2);
\end{pgfonlayer}
}

\newcommand{\ArrowL}[3]%
{ \begin{pgfonlayer}{background}
    \draw[->,#3] (#1) to[bend left=45] (#2);
  \end{pgfonlayer}
}

\newcommand{\EdgeL}[3]%
{ \begin{pgfonlayer}{background}
\draw[dashed,#3] (#1) to[bend right=-45] (#2);
\end{pgfonlayer}
}

\newcommand{\Arrow}[3]%
{ \begin{pgfonlayer}{background}
\draw[->,#3] (#1) -- +(#2);
\end{pgfonlayer}
}

\allowdisplaybreaks
\pgfarrowsdeclare{arcs}{arcs}{...}
{
  \pgfsetdash{}{0pt} 
  \pgfsetroundjoin   
  \pgfsetroundcap    
  \pgfpathmoveto{\pgfpoint{-5pt}{5pt}}
  \pgfpatharc{180}{270}{5pt}
  \pgfpatharc{90}{180}{5pt}
  \pgfusepathqstroke
}
\newcommand{\ArrowB}[3]%
{ \begin{pgfonlayer}{background}
    \draw[|-arcs,line width=0.4mm,shorten <= 0.3cm,shorten >= 0.3cm,#3] (#1) -- +(#2);
  \end{pgfonlayer}
}

\newcommand{\titlepaper}{More efficient transportation of causal effects using mediators}

\date{\today}

\author[1]{Xinyi Zhang\thanks{Corresponding author: \href{mailto:zhang.xinyi@nyu.edu}{zhang.xinyi@nyu.edu}}}
 \author[2]{Kara E. Rudolph}
 \author[3]{James Russell Pike}
 \author[1]{Richard Liu}
 \author[1]{Michele Santacatterina}
 \author[1]{Iv\'an D\'iaz}
 \affil[1]{\footnotesize
 Division of Biostatistics, Department of Population Health, New York University Grossman School of Medicine, New York, New York, USA.}
 \affil[2]{\footnotesize
 Department of Epidemiology, Mailman School of Public Health, Columbia University, New York, New York, USA.}
 \affil[3]{\footnotesize
 Optimal Aging Institute, New York University Grossman School of Medicine, New York, New York, USA.}

\title{\titlepaper}

\begin{document}
\maketitle

\begin{abstract}
We develop flexible semiparametric estimators of the average treatment effect transported to a target population that leverage mediators to achieve identification, and prior knowledge of population heterogeneity or summary-level statistics to achieve improved efficiency. We first propose a one-step semiparametric estimator that assumes knowledge of the baseline covariates driving the heterogeneity across populations. This approach improves efficiency by transporting only through these covariates, rather than through the full covariate set. We then introduce a collaborative one-step semiparametric estimator that also improves efficiency without requiring knowledge of which covariates drive study heterogeneity. Collectively, these estimators advance our ability to transport causal effects by improving efficiency in the presence of mediation and intermediate confounding through targeted dimension reduction strategies. We use simulation to examine finite-sample performance and apply our estimators to a large-scale housing mobility trial (MTO) and a hearing intervention trial (ACHIEVE).
\end{abstract}

\section{Introduction}
Causal effects often differ across populations for reasons that include, but are not limited to, variation across populations in effect-modifying covariates, in treatment uptake or adherence, in the mechanisms generating key mediators, and in the processes determining outcomes. Such variation poses fundamental challenges for the external validity of causal effect estimates \citep{rothwell2005external, pearl2022external, rudolph2022efficiently}. 
To address this issue, researchers have developed statistical estimators that combine information from source and target populations to estimate the transported causal effects, facilitating replication studies and enhancing the generalizability of source population findings \citep{cole2010generalizing, stuart2011use, tipton2013improving, buchanan2018generalizing, dahabreh2019generalizing}. For example, \citet{cole2010generalizing} demonstrated that standardizing the results of an HIV trial to a target population produces transported effects that better reflect population-level impacts, highlighting the value of transport methods in improving the generalizability of the study.\\

The vast majority of methods for transport assume that differences in the distribution of effect modifiers between populations are the only drivers of differences in average treatment effects between populations. However, mediators whose distribution differs across populations may be another important driver for different treatment effects across populations. For example, when treatment influences intermediate variables that vary across populations, such as treatment uptake, transportability requires careful consideration. This issue has been addressed in the context of encouragement-design trials, where treatment is randomized but actual exposure to the treatment may differ from assignment \citep{rudolph2017robust, rudolph2020using}. Previous work developed doubly robust and efficient estimators for transported total, direct, and indirect effects, accounting for post-treatment intermediate variables \citep{rudolph2021transporting, rudolph2022efficiently}. Similarly, \citet{dahabreh2022generalizing} discussed a framework for transporting effects from a randomized trial to a non-randomized study in the presence of treatment nonadherence under the assumptions that there is no effect of trial participation on the outcome and no direct effect of assignment on outcome. More recently, \citet{pmlr-v275-mhasawade25a} examined how different missing data mechanisms in mediators affect transportability. \citet{li2026semiparametric} introduced a data fusion framework for mediation analysis when the mediator and outcome are never jointly observed, leveraging shared instrumental variables to identify natural direct and indirect effects and develop robust estimators. Despite these advances, existing estimators for transported effects that utilize intermediate, mediating variables---either for identification of total effects or because a mediational estimand is of interest---often do not fully exploit opportunities for efficiency gains, particularly in settings where mediators are affected by intermediate confounders that jointly influence downstream outcomes---a challenge that has been emphasized in prior work on mediation analysis \citep{avin2005identifiability, andrews2021insights, diaz2024non, vo2024recanting, liu2024general}. \\

Limited overlap in the distributions of baseline covariates and mediating variables between populations is one factor that can contribute to increased variance in semiparametric transport estimators, motivating the development of more efficient estimation methods \citep{rudolph2025improving}. In this paper, we show that the precision of transported estimates in the presence of mediators can be improved through two forms of covariate dimension reduction, extending methods developed in \citet{rudolph2025improving}. First, we develop a semiparametric estimator tailored to settings in which the variables driving population heterogeneity are known. The key insight is that, for transportability, it suffices to transport a conditional average treatment effect (CATE) that marginalizes all other baseline covariates except those informative for baseline population heterogeneity. In the case where prior knowledge of covariates that drive population heterogeneity is unknown, we propose a collaborative efficient estimator that can further improve efficiency without such information \citep{benkeser2020nonparametric, van2010collaborative, rudolph2025improving}. This method draws on the following important ideas: (1) the fundamental role of the propensity score---in this case, the population selection score \citep{rosenbaum1983central, tipton2013improving}, (2) the probability of being in the source population is conditional on a dimension reduction summary statistics instead of all covariates, and (3) the outcome model and selection score act collaboratively on estimation. This estimator avoids requiring knowledge of the variables driving population heterogeneity, at the price of not being doubly robust by default, and relies on a consistent estimate of the target population CATE.\\

The estimation strategy of our first proposed estimator relies on correcting the first-order bias of plug-in machine learning estimators using the \textit{efficient influence function} (EIF) \citep{vanderVaart98}. The EIF plays a central role in semiparametric efficiency theory because its variance defines the semiparametric efficiency bound \citep{Bickel97}, and it characterizes the leading term in the von Mises expansion of a plug-in estimator's error \citep{mises1947asymptotic}, therefore having an estimation error that is second-order and may be negligible. Our first contribution is to obtain the EIF and the corresponding bias-corrected estimator under the assumption that a known transformation of covariates varies in distribution across the source and target population. Our second contribution is to propose another estimation strategy when such transformation is unknown. Inspired by \citet{benkeser2020nonparametric}, our estimation strategy relies on an alternative von Mises expansion that does not lead to the EIF but nonetheless yields estimators with second-order error that may be negligible under some assumptions. Our proposed estimators improve previous work by improving efficiency in estimating transported causal effects when mediators and intermediate confounders are present. \\

The remainder of the paper is organized as follows. In Section \ref{sec:causal_model}, we introduce the structural causal model and discuss identification strategies. Section \ref{sec:eff_theory} presents the efficiency theory underlying our proposed estimators. In Section \ref{sec:estimator}, we provide details on estimation and discuss the robustness properties of the proposed estimators. Section \ref{sec:simulation} presents simulation studies, and Section \ref{sec:application} demonstrates the application of our methods using data from the Moving to Opportunity Study (MTO)---a randomized trial where families living in public housing in five U.S. cities could receive a Section 8 voucher for subsidized rent on the private market. Our goal is to assess the impact of randomized assignment to receive a Section 8 housing voucher or not on the risk of developing a mood disorder at the 10--15 year follow-up among participants' sons, and to transport the estimated effects from New York City and Chicago to Los Angeles. In addition, we apply our proposed methods to the Aging and Cognitive Health Evaluation in Elders (ACHIEVE) study, a randomized trial evaluating the effect of a hearing intervention on cognitive outcomes among older adults. We assess transporting the treatment effect from the Atherosclerosis Risk in Communities (ARIC) subgroup of ACHIEVE, consisting of participants recruited from the ARIC cohort, to the de novo subgroup of ACHIEVE, consisting of volunteers newly recruited for the ACHIEVE study.

\section{Causal model and identification}
\label{sec:causal_model}

\subsection{Preliminaries on general transportability}
In the absence of mediators, the observed data for transportability \citep{pearl2015generalizing} take the form $O = (S, W, A, S \times Y)$. Here, $S$ is a population indicator, with $ S = 1 $ denoting the source population (where outcomes are observed) and $ S = 0 $ denoting the target population (where outcomes are unobserved). $ W $ represents a vector of baseline covariates. $ A $ denotes the binary treatment or exposure. $ Y $ is the outcome observed when $S = 1$. Let $Y(a)$ denote the counterfactual outcome corresponding to the setting where $A = a$. Additional notation will be introduced later, but Table \ref{tab:notation} provides a summary of notation in this paper.\\

\begin{table}[ht]
\centering
\caption{Notation\label{tab:notation}}
{\scriptsize  
\begin{tabular}{ll}
\toprule
\textbf{Symbol} & \textbf{Definition} \\
\midrule
$S$ & Population indicator (1 = source, 0 = target) \\
$W$ & Baseline covariates \\
$A$ & Treatment indicator \\
$Y$ & Outcome \\
$M$ & Mediator \\
$Z$ & Intermediate confounder \\
$X$ & Collection of $(M,Z)$ \\
$O$ & Observed data: $(S, W, A, M, Z, S \times Y)$ \\
$U$ & Unmeasured exogenous errors \\
$\P$ & Distribution of observed data $O$ \\
$V$ & Partial study heterogeneity \\
$h(a,w)$ & Mean outcome in the target population given $A = a$ and $W = w$: $\E( Y \mid A = a, W = w, S = 0)$ \\
$g(W)$ & Mean outcome in the target population given $A = 0$:  $h(0,W)$ \\
$f(W)$ & CATE in the target population: $h(1,W) - h(0,W)$ \\

$m(A,W,X)$ & Outcome model in the source population: $\E(Y\mid A, W, X, S=1)$  \\

$e_A(W)$ & Propensity score for being treated in target population: $\P(A=1\mid W, S=0)$ \\

$e_S(W)$ & Selection score for being in the source population: $\P(S = 1 \mid W)$ \\

$q_S(A,W,X)$ & Selection score for being in source population given all variables: $\P(S = 1 \mid A, W, X)$ \\

$r_S(A,W,X)$ & Selection score for being in source population conditional on the outcome model: $\P(S=1 \mid m(A, W, X))$\\

$\p$ & Proportion in target population: $\P(S = 0)$ \\

$p_X(a, w)$ & Probability that $X = x$ in the target population conditional on $A = a$ and $W = w$: $\P(X = x \mid A = a, W = w, S = 0)$ \\

$l(V)$ & Conditional expectation of CATE on $V$: $\E\{f(W)\mid V\}$ \\

$k(W)$ & Conditional expectation of CATE on selection score of being in the source population: $\E\{f(W)\mid e_S(W)\}$ \\

$\theta$ & Transported ATE under assumptions \ref{ass:A3} and \ref{ass:positivity} \\
$\lambda$ & Transported ATE under assumptions \ref{ass:A3}, \ref{ass:positivity}, and \ref{ass:modif} \\
$\tau$ & Transported ATE under assumptions \ref{ass:A3} and \ref{ass:positivity}, expressed through the selection score: $\E[\E\{f(W) \mid e_S(W)\} \mid S=0]$ \\
$\tilde\theta$ & One-step estimator for $\theta$, based on the EIF developed in \citet{rudolph2017robust} \\
$\tilde\lambda$ & One-step estimator for $\lambda$ \\
$\tilde\tau$ & Collaborative one-step estimator for $\tau$ \\

$\D_\theta(O; \xi)$ & EIF for $\theta$ \\
$\D_\lambda(O; \eta)$ & EIF for $\lambda$ \\
$\U_\tau(O; \gamma)$ & Estimating function for $\tau$\\
$\D_{h,a}(O;\kappa)$ & Doubly robust pseudo-outcome for $h(a,W)$, $a\in\{0,1\}$ \\
$\xi$ & Nuisance parameters for $\D_\theta(O; \xi)$ \\
$\eta$ & Nuisance parameters for $\D_\lambda(O; \eta)$ \\
$\gamma$ & Nuisance parameters for $\U_\tau(O; \gamma)$ \\
$\kappa$ & Nuisance parameters for $\D_{h_{a}}(O; \kappa)$\\
\bottomrule
\end{tabular}
}
\end{table}

Two typical assumptions invoked by transport methods are the exchangeability of the outcome mechanism between the source and the target population, along with the assumption that there are no unmeasured treatment-outcome confounders in the target population---that usually would hold by design if the target population is randomized:

\begin{assumptioniden}[No selection or confounding bias conditional on baseline variables - observed-outcome version]\label{ass:trans1}
Assume that
\begin{enumerate}[label=(\roman*)]
    \item (S - admissibility). $\E(Y\mid A, W, S=1) =\E(Y\mid A, W, S=0) $, and that\label{ass:trans1_exch}
    \item  (Conditional exchangeability on $S = 0$). $Y(a)\indep A\mid W, S=0$.\label{ass:trans1onf}
\end{enumerate}
\end{assumptioniden}
Under \ref{ass:trans1} and appropriate positivity conditions listed in \citet{rudolph2025improving}, the transported average treatment effect $\E[Y(1) - Y(0)\mid S=0]$ can be identified as $\E\{\E(Y\mid A=1, W, S=1) - \E(Y\mid A=0, W, S=1)\mid S=0\}.$

\begin{remark}
    In cases where the source but not the target population is randomized, a more commonly known identification result states that the above parameter identifies the transported effect under the following assumption.
\begin{assumptioniden}[No selection or confounding bias conditional on baseline variables - counterfactual version]\label{ass:trans2} Assume that
\begin{enumerate}[label=(\roman*)]
          \item (S - admissibility). $\E[Y(a)\mid W, S=0]=\E[Y(a)\mid W, S=1]$, and that\label{ass:trans2_exch}
    \item  (Conditional exchangeability on $S = 1$). $Y(a)\indep A\mid W, S=1$.\label{ass:trans2onf}
\end{enumerate}
\end{assumptioniden}
If both studies have measured outcomes, \ref{ass:trans1}\ref{ass:trans1_exch} is testable, whereas  \ref{ass:trans2}\ref{ass:trans2_exch} is not. When both studies are randomized, \ref{ass:trans1} and \ref{ass:trans2} are equivalent.
\end{remark}

While widely used for general transported causal effect estimation, these assumptions may fail when the population indicator $S$ affects intermediate variables on the pathway from $A$ to $Y$, such as treatment adherence. In the next section, we present a structural causal model for transportability under mediator $M$ and intermediate confounder $Z$, along with the required identification assumptions.

\subsection{Structural causal model for transportability under mediation and intermediate confounding}

Assumptions \ref{ass:trans1}\ref{ass:trans1_exch} and \ref{ass:trans2}\ref{ass:trans2_exch} may be violated when $S$ affects variables in the causal pathway from $A$ to $Y$, even if there is no direct effect of $S$ on $Y$. This situation is depicted in the directed acyclic graph (DAG) in Figure~\ref{fig:dag}. 
\begin{center}
\begin{minipage}{0.45\textwidth}
\small
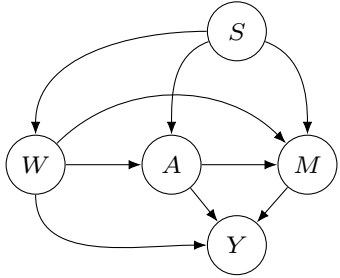
\begin{figure}[H]
  \centering
 \begin{tikzpicture}[node distance =0.85 cm and 0.85 cm]
    \node (a) [free]{$A$}; 
    \node (w) [left = 10mm of a, free]{$W$};
    \node (m) [right = 10mm of a, free]{$M$};
    \node (y) [below right= 5mm and 3mm of a, free]{$Y$};   
    \node (s) [above right = 12mm and 3mm of a, free]{$S$};   
    \path (w) edge node[above, el] {} (a);
    \path (a) edge node[above, el] {} (m);
    \path (m) edge node[above, el] {} (y);
    \path (w) edge[out=45, in=135] node[above, el] {} (m);    
    \path (w) edge[out=-90,in=180] node[above, el] {} (y);
    \path (s) edge[out=180,in=90] node[above, el] {} (w);
    \path (s) edge[out=200,in=90] node[above, el] {} (a);
    \path (a) edge node[above, el] {} (y);
    \path (s) edge[out=-20,in=90] node[above, el] {} (m);
\end{tikzpicture}
  \caption{Causal selection diagram under mediation.}
  \label{fig:dag}  
\end{figure}
\end{minipage}%
\hspace{6mm}
\begin{minipage}{0.45\textwidth}
\small
 \begin{figure}[H]
  \centering
 \begin{tikzpicture}[node distance =0.85 cm and 0.85 cm]
    \node (a) [free]{$A$}; 
    \node (w) [left = 10mm of a, free]{$W$};
    \node (m) [right = 10mm of a, free]{$M$};
    \node (y) [below right= 5mm and 3mm of a, free]{$Y$};   
    \node (z) [right = 10mm of m, free]{$Z$};
    \node (s) [above right = 12mm and 3mm of a, free]{$S$};   

    \path (w) edge node[above, el] {} (a);
    \path (w) edge[out=45,in=135] node[above, el] {} (z);
    \path (a) edge node[above, el] {} (m);
    \path (m) edge node[above, el] {} (y);
    \path (a) edge[out=45,in=145] node[above, el] {} (z);
    \path (z) edge node[above, el] {} (m);
    \path (z) edge[out=-90,in=0] node[above, el] {} (y);
    \path (w) edge[out=45, in=135] node[above, el] {} (m);    
    \path (w) edge[out=-90,in=180] node[above, el] {} (y);
    \path (s) edge[out=180,in=90] node[above, el] {} (w);
    \path (s) edge[out=200,in=90] node[above, el] {} (a);
    \path (a) edge node[above, el] {} (y);
    \path (s) edge[out=-20,in=90] node[above, el] {} (m);

\end{tikzpicture}
  \caption{\small Causal selection diagram under mediation and intermediate confounding.}
  \label{fig:dag_ic}  
\end{figure}
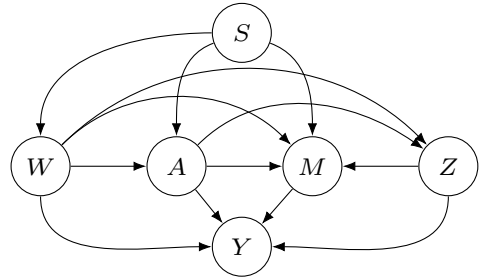
\end{minipage}%
\end{center}

Under this model, transportability is possible when we have mean-exchangeability of the outcome, adjusting for baseline covariates and the mediator \citep{rudolph2017robust}:
\begin{equation}
    \E(Y\mid A, W, M, S=1) = \E(Y\mid A, W, M, S=0).\label{eq:ass}
\end{equation} 

However, in the presence of a confounder between $M$ and $Y$, such as the variable $Z$ depicted in Figure~\ref{fig:dag_ic}, transporting effects using intermediate variables requires additional care, even in a case where the intermediate variable $Z$ is not affected by the study indicator $S$. Specifically, (\ref{eq:ass}) fails to hold in the model depicted in Figure~\ref{fig:dag_ic}. This is because only conditional on $M$, which is a collider between $S$ and $Z$, can induce spurious associations between these two variables and open a backdoor path $S \to Z\to Y$ which would otherwise be closed. \\

To account for such an intermediate confounder $Z$, we assume that the observed data is $O = (S,W,A,M,Z,S \times Y)$ and assume that $(O_{1},...,O_{n})$ represents a sample of i.i.d. observations of $O$. We assume the following structural model: 
$S = f_S(U_S); W = f_W(S, U_W); A = f_A(S, W, U_A); Z = f_Z(A,W,U_Z); M = f_M(A,W,S,Z,U_M); Y = f_Y(A,W,M,Z,U_Y)$, corresponding to the DAG in Figure~\ref{fig:dag_ic}. Under this model, each variable is assumed to be generated as a deterministic, unknown function of exogenous error terms 
$(U_S, U_W, U_A, U_Z, U_M, U_Y) $ and previously observed endogenous variables. This framework follows the structural causal model formalism introduced by \citet{pearl1995causal}. \\

We introduce the following identification assumptions: 
\begin{assumptioniden}[No selection or confounding bias conditional on baseline and intermediate variables - observed-outcome version]\label{ass:A3} Assume that
\leavevmode
\begin{enumerate}[label=(\roman*)]
    \item (S-admissibility). $ \E(Y \mid A, W, Z, M, S=1) = \E(Y \mid A, W, Z, M, S=0) $\label{ass:A3_S}, and that
    \item (Conditional exchangeability on $S = 0$). $ Y(a) \indep A \mid W, S=0.$ \label{ass:A3_exch}
\end{enumerate}
\end{assumptioniden}

Additionally, to have the transported treatment effects well-defined and identifiable, we introduce the following positivity assumptions:
\begin{assumptioniden}[Positivity of study and treatment mechanisms]\label{ass:positivity} Assume that
\leavevmode
\begin{enumerate}[label=(\roman*)]
    \item $\P(W=w\mid S=0)>0$ implies $\P(A=a\mid W=w, S=0)>0$, and that \label{ass:pos_A}
    \item  $\P(M = m, Z = z \mid A = a, W = w, S = 0) \P(W = w \mid S = 0) > 0$ implies $\P(M = m, Z = z, A = a, W = w, S = 1) > 0.$  \label{ass:pos_X}
\end{enumerate}
\end{assumptioniden}




After introducing the intermediate variables $M$ and $Z$, our new assumption \ref{ass:A3} is a variation of \ref{ass:trans1}. In addition, the identification results below would be correct if we were to assume a variation of \ref{ass:trans2}.

\begin{assumptioniden}[No selection or confounding bias conditional on baseline and intermediate variables - counterfactual version]\label{ass:A5} Assume that
\leavevmode
\begin{enumerate}[label=(\roman*)]
    \item (S-admissibility). $ \E(Y(a) \mid W, M(a), Z(a), S=1) = \E(Y(a) \mid W, M(a), Z(a), S=0) $, and that
    \label{ass:A5_S}
    \item (Conditional exchangeability). 
     $Y(a)\perp\!\!\!\perp A \mid W, M(a), Z(a), S=1$ and $(M(a), Z(a))\perp\!\!\!\perp A \mid W, S=0$.
     \label{ass:A5_exch}
\end{enumerate}
\end{assumptioniden}Note that even though the conditional exchangeability for the outcome $Y(a)$ is assumed conditional on $S = 1$, we do need to assume the conditional exchangeability for $(M(a), Z(a))$ on $S = 0$. Therefore, \ref{ass:A5} may not be practically as useful as \ref{ass:A3}, as \ref{ass:A5} imposes a type of conditional exchangeability assumption in both populations, whereas \ref{ass:A3} requires conditional exchangeability only on the target. That said, the identification results we present next hold under either set of assumptions. \\

For identification under assumption \ref{ass:A3} \& \ref{ass:positivity}, we begin by defining the following parameter:
\begin{equation*}
    h(a,w)=\int \E(Y\mid A=a,  W=w, Z=z, M=m, S=1)\dd\P(m,z\mid A=a, W=w, S=0).
\end{equation*} The quantity $h(a,w)$ represents the expected outcome for individuals with covariates $W=w$ under treatment $A=a$, where the outcome model is learned from the source population $(S = 1)$, but the distribution of the intermediate variables $X = (M, Z)$ is learned from the target population $(S = 0)$. It will be convenient in our developments to re-parameterize this function as $h(a,w)=a\times f(w) + g(w)$ where $f(w)=h(1,w)-h(0,w)$ and $g(w)=h(0,w)$ \citep{chernozhukov2018double, hahn2020bayesian, nie2021quasi}. As shown in the first section of Appendix \ref{subsec:identification}, under assumptions \ref{ass:A3} \& \ref{ass:positivity}, the quantity $h(a,w)$ simplifies to $\E[Y\mid A=a, W=w, S=0]$ and $f(W)$ is the CATE in the target population. \citet{rudolph2017robust} showed that the transported treatment effect can be identified as
\[\E[Y(1) - Y(0)\mid S=0] = \E[f(W)\mid S=0] =: \theta.\] Assuming \ref{ass:A5} instead of \ref{ass:A3} would lead to the same identification results for, as shown in Appendix \ref{subsec:identificationA5}. The EIF for $\theta$ in the non-parametric model was given by  \citet{rudolph2017robust} as:
{\small
\begin{align*}
    \D_\theta(O;\P) = \frac{1}{\P(S=0)}\bigg[&\frac{\one\{S=1\}(2A-1)}{\P(A\mid W,S=0)}\frac{\P(S=0\mid A, W, X)}{\P(S=1\mid A, W, X)}\times\{Y-\E(Y\mid A, W, X, S=1)\}\\
    &+\frac{\one\{S=0\}(2A-1)}{\P(A\mid W,S=0)}\times\{\E(Y\mid A, W, X, S=1)-h(A,W)\}\\
               &+\one\{S=0\}[f(W)-\theta]\bigg],
\end{align*}}%
where we have simplified the notation by defining $X=(M, Z)$. Note that $\D_\theta$ depends on $\P$ only through a vector of nuisance parameters $\xi(\P)$; therefore, we will sometimes use the notation $\D_\theta(O;\xi)$. This EIF characterizes the construction of a one-step estimator by correcting the first-order bias from the plug-in estimator, where the plug-in estimator is denoted as \[\hat\theta = \frac{1}{n\times \hat\P(S = 0)}\sum_{i=1}^n\one\{S_i=0\}\hat f(W_i).\]
Thus, the one-step estimator \citep{kennedy2024semiparametric} that follows from the EIF can be expressed as $\tilde{\theta} = \hat{\theta} + \frac{1}{n} \sum_{i=1}^{n} \D_{\theta}(O_i; \hat{\xi})$. We consider the existing estimator $\tilde{\theta}$ as our baseline estimator throughout the paper.

\subsection{Leveraging known study heterogeneity for improved efficiency}
\label{subsec:study}

Estimators that use a subset of effect modifiers required for transport to obtain efficiency gains without mediators were developed in \citet{rudolph2025improving}. In this section, we focus on variables driving between-population heterogeneity, regardless of whether they are known effect modifiers or not. Specifically, our approach uses the fact that only a transformation of baseline covariates that captures between-study heterogeneity is needed for transporting causal effects. We formally state our main assumption as: 

\begin{assumptioniden}[Known partial study heterogeneity]\label{ass:modif}
    Assume there is a known transformation $V=V(W)$ such that $S\indep W\mid V$.
\end{assumptioniden}

This assumption states that there exists a transformation $V(W)$  (e.g., a subset) that is sufficient to explain differences across study populations. Importantly, this assumption is testable, and such a transformation can, in principle, be learned from data. Next, we introduce an identification parameter $\lambda$ that leverages this assumption to improve efficiency in transporting treatment effects using mediators and intermediate confounders.  
\begin{proposition}\label{prop:identification}
    Assume \ref{ass:A3}, \ref{ass:positivity}, and \ref{ass:modif}, the average treatment effect in the target population $\E[Y(1)-Y(0)\mid S=0]$ is identified as
    \[\lambda := \E\{\E[f(W)\mid V]\mid S=0\}.\]
\end{proposition}

 We include the proof in Appendix \ref{subsec:identification}. Before proceeding to the development of estimators for $\lambda$, we would like to introduce an alternative identification result that will allow us to develop estimators with efficiency gains in situations where the transformation $V$ is unknown or \ref{ass:modif} cannot be safely assumed. 

 \subsection{Leveraging the fundamental role of selection score for improved efficiency} \label{subsec:selection_score}

In this section, we extend the collaborative estimator in  \citet{rudolph2025improving} to the setting with intermediate variables, again leveraging the 
fundamental role of the propensity score---in this case, the selection score. This approach enables efficiency gains without requiring explicit modeling of partial study heterogeneity or reliance on Assumption~\ref{ass:modif}.\\

We let $e_S(W) = \P(S = 1 \mid W)$ denote the probability of being in the source population conditional on covariates. This quantity has also been referred to as \textit{selection score} in prior literature \citep{tipton2013improving}. As in \citet{rudolph2025improving}, we can re-parameterize $\theta$ as a function of the selection score $e_S(W)$ given $S \indep W \mid e_S(W)$ \citep{rosenbaum1983central}:
\begin{equation}
    \E\left[f(W) | S = 0\right] = \E\left[\E\left[f(W) | e_S(W), S = 0 \right] | S = 0\right] = \E\left[\E\left[f(W) | e_S(W) \right] | S = 0\right] =: \tau \label{eq:theta_transform}.
\end{equation}
 As shown in Equation \ref{eq:theta_transform}, the transformation of $\theta$ under the fundamental role of the selection score allows us to have $\tau$ as our new estimand for transported causal effects. Intuitively, we can regard $e_S(W)$ as an unknown form of $V$ described in the previous Section~\ref{subsec:study}. Note that $e_S(W)$ provides a sufficient dimension reduction of the covariates $W$ for transport. Thus, once $f(W)$ has been estimated, only its relationship with the one-dimensional score $e_S(W)$ must be learned in order to transport the treatment effect. In Section~\ref{subsec:eff_tau}, we show how the selection score can be further leveraged to improve efficiency through its collaboration with the outcome regression.

\section{Preliminary theory for estimation}
\label{sec:eff_theory}

Our estimation focuses on obtaining the one-step estimators by correcting the first-order bias of the plug-in estimators, where such methods have been recently adapted in the literature such as targeted learning \citep{vanderLaanRose11, vanderLaanRose18} and double machine learning \citep{chernozhukov2018double, kennedy2024semiparametric}. This strategy relies on characterizing the first-order bias of the plug-in estimator through the von Mises expansion \citep{mises1947asymptotic} and is defined as follows: 

\begin{definition}[First-order von Mises expansion]
Suppose $\psi: \mathcal{P} \to \mathbb{R}$ is a functional that is smooth and maps from the nonparametric model to the real line. The functional satisfies the \textit{first-order von Mises expansion} if the following holds
\begin{align*}
\psi(\F) - \psi(\P)
= - \E_\P\!\left[ \varphi(O; \F) \right] + R_2(\F,\P) 
\end{align*}
for any distributions $\F, \P \in \mathcal{P}$, where  $\varphi(O; \F)$ is a mean zero, finite-variance function satisfying  $\E_\F\!\left[\varphi(O; \F)\right] = 0$ and $\E_\F\!\left[\varphi^2(O; \F)\right] < \infty$ and $R_2(\F,\P)$ is a second-order term given only by products or squares of differences between 
 $\F$ and $\P$. The function $\varphi$ is referred to as a
gradient of the first-order expansion.
\end{definition}

Given the true data distribution $\P$, let $\hat{\P}$ be a data distribution estimated from the distribution of the observed data. The von Mises expansion defined above allows us to characterize the bias of the plug-in estimator $\psi(\hat{\P})$, which is given by
\begin{align*}
    \psi(\hat{\P}) - \psi(\P)
= - \E_\P\![ \varphi(O; \hat{\P})] + R_2(\hat{\P},\P).
\end{align*}

If we have a sufficiently accurate estimate of $\P$, the second-order term $R_2(\hat{\P}, \P)$ is small and therefore negligible. In this case, the bias of the plug-in estimator is dominated by the term $\E_{\P}[ \varphi(O; \hat{\P})]$. This term can be estimated by the empirical mean of $\varphi(O_i; \hat{\P})$ and added back to the plug-in estimator to correct its bias. Consequently, the one-step estimator is given by $\psi(\hat{\P}) + \frac{1}{n}\sum_i\varphi(O_i; \hat{\P}).$ We note that when $\varphi$ is the canonical gradient of $\psi$, it corresponds to the EIF, which achieves the minimum variance among regular influence functions. This decomposition motivates bias-corrected and one-step estimators that remove the first-order error of the plug-in estimator and attain semiparametric efficiency under appropriate regularity conditions. 

\subsection{First-order bias of plug-in estimators for $\lambda$ and efficiency considerations}

To develop an estimator for $\lambda$, we introduce the following notation: 
\begin{align*}
    m(a,w,x)&=\E(Y\mid A=a, W=w, X=x, S=1)\\
    e_A(w)&=\P(A=1\mid W=w, S=0)\\
    e_S(v)&=\P(S=1\mid V=v)\\
    q_S(a,w,x)&=\P(S=1\mid A=a, W=w, X=x)\\
    \p &=\P(S=0)\\
    l(v)&=\E\{f(W)\mid V=v\}.
\end{align*} 

With this notation, we can express the target estimand $\lambda = \E\{\E[f(W)\mid V]\mid S=0\}$ as $\lambda = \E[l(V) \mid S = 0 ]$. The following theorem gives the EIF for $\lambda$.

\begin{theorem} \label{theorem:EIF_lambda}
Let $\eta = (m, f, g, e_A, e_S, q_S, \p, l)$ denote the collection of relevant nuisance functions. Under Assumption~\ref{ass:modif}, the EIF for $\lambda$, denoted by $\D_\lambda(O;\eta)$, is given by
{\small
\begin{align*}
    \D_\lambda(O;\eta) = \frac{1}{\p}\bigg[&\frac{S\{A-e_A(W)\}}{e_A(W)\{1-e_A(W)\}}\frac{1-q_S(A, W, X)}{q_S(A, W, X)}\times\{Y-m(A, W, X)\}\\
    &+\frac{(1-S)\{A-e_A(W)\}}{e_A(W)\{1-e_A(W)\}}\times\{m(A,W,X)- h(A,W)\}\\
    &+\{1-e_S(V)\}\times \{f(W) - l(V)\}\\
               &+(1-S)\{l(V)-\lambda\}\bigg].
\end{align*}}
\end{theorem}

Note that the first two terms in $\D_\lambda(O;\eta)$ mirror those in $\D_\theta(O;\xi)$; the key distinction lies in the final two terms. In particular, $\D_\lambda(O;\eta)$ uses $l(V)=\E\{f(W)\mid V\}$, so that transport depends only on the lower-dimensional summary $V$ rather than the full covariate vector $W$.

\begin{theorem}[First-order bias of plug-in estimator for $\lambda$] \label{theorem:lambda}
For any estimator $\hat\eta=(\hat{m}, \hat f, \hat g,  \hat e_A, \hat e_S, \hat q_S, \hat \p, \hat l)$, 
let $\hat\lambda$ denote the plug-in estimator defined as 
\[\hat\lambda = \frac{1}{n\times \hat \p}\sum_{i=1}^n\one\{S_i=0\}\hat l(V_i).\]
Assume that for some $\delta > 0$, we have $0 < q_S(A,W,X) < 1$ and $0 < e_A(W) < 1$, and the estimators satisfy $\P(\delta<\hat q_S(A,W,X)<1-\delta) = 1$ and
$\P(\delta<\hat e_A(W)<1-\delta) = 1$. The bias of $\hat\lambda$ is equal to
$\hat\lambda - \lambda = -\E[\D_{\lambda}(O;\hat \eta)] + R_\D(\eta, \hat \eta),$ where $R_\D$ is a second-order term (defined in the proof of the theorem) equal to a sum of products of errors of the type $\int c(\eta,\hat\eta)(a(\hat\eta) - a(\eta))(b(\hat\eta) - b(\eta))\dd\P$ and is upper bounded as
{\footnotesize
\[
R_\D(\eta, \hat \eta) \lesssim 
\| \hat{q}_S - q_S \| \cdot \| \hat{m} - m \| 
+ \| \hat{e}_S - e_S \| \cdot \| \hat{l} - l \| 
+ \| \hat{e}_S - e_S \| \cdot \| \hat{f} - f \|\] 
\[
+ \| \hat{e}_A - e_A \| \cdot \| \hat{f} - f \| 
+ \| \hat{e}_A - e_A \| \cdot \| \hat{g} - g \| 
+ (\hat{\p} - \p)(\hat{\lambda} - \lambda),
\]
}

where $\lesssim$ denotes less than or equal to up to some constants; the proof is included in the Appendix \ref{subsec:vonmise_lambda}. 
\end{theorem}

Under suitable convergence-rate conditions on the nuisance estimators, the one-step estimator is asymptotically normal with asymptotic variance equal to the variance of its EIF, as shown in Section \ref{sec:estimator}. Therefore, comparing $\var[\D_\theta(O;\xi)]$ and $\var[\D_\lambda(O;\eta)]$ clarifies the efficiency gains obtained by using the identification result for $\lambda$ instead of $\theta$. We present this comparison in the following lemma.

\begin{lemma}[Efficiency gains] \label{lemma:gains}
    Under Assumption \ref{ass:modif}, the difference between variances of $\D_\lambda(O;\eta)$ and $\D_\theta(O;\xi)$ is equal to:
{\footnotesize\[\var[\D_\lambda(O;\eta)]-\var[\D_\theta(O;\xi)]=\frac{1}{\P(S = 0)^2}\E\bigg\{\P(S=0\mid V)\{\P(S=0\mid V)-1\}\big[f(W)-\E\{f(W)\mid V\}\big]^2\bigg\}\leq 0,\]}%
showing that using $\lambda$ instead of $\theta$ leads to at least the same or more efficiency. A proof for this lemma is in Appendix \ref{subsec:eff_gain}. 
\end{lemma}

Although efficiency gains can be demonstrated by comparing the variances of the influence functions for the two estimators, the intuition behind the improved efficiency of using $\lambda$ instead of $\theta$ to transport causal effects lies in the identification result established in the proof of Proposition \ref{prop:identification}. Specifically, using Assumption~\ref{ass:modif} for $\lambda$, one can marginalize out all other variables in the CATE before transporting (i.e., compute $\E[ f(W) \mid V ]$), and then transport using only $V$. In contrast, estimating $\theta$ generally requires transporting treatment effect heterogeneity conditional on the full set of covariates and treatment mechanisms. This reduction in the information required for estimation explains the potential efficiency gain when targeting $\lambda$. \\

According to the difference in variances provided in Lemma \ref{lemma:gains}, greater efficiency gains are obtained when the quantity $[f(W) - \E\{f(W)\mid V\}]^2$ is large and the quantity $\P(S=0\mid V)\{\P(S=0\mid V) - 1\}$ is small. The largest efficiency gains are therefore realized when $ V $ is as uninformative about $ f(W) $ as possible, and when $ \P(S = 0 \mid V) = 0.5 $. When $ V $ is uninformative about $ f(W) $, the predicted quantity $ \E\{f(W) \mid V\} $ would be closer to the marginal average of $f(W)$, thus the squared error $[f(W) - \E\{f(W)\mid V\}]^2$ is greater.  We examine numerical scenarios with efficiency gain in Section \ref{sec:simulation}. 

\subsection{First-order bias of plug-in estimators for $\tau$}
\label{subsec:eff_tau}
We now construct an estimator of $\tau$ that we refer to as a \emph{collaborative efficient estimator} \citep{benkeser2020nonparametric, rudolph2025improving}. The estimator is collaborative for two reasons. First, it projects the CATE onto the selection score as discussed in Section \ref{subsec:selection_score}, replacing adjustment on the full covariate vector $W$ with adjustment on the lower-dimensional summary $e_S(W)$. Second, it replaces the full study-selection model $q_S(A,W,X)=\P(S=1\mid A,W,X)$ with the reduced model $r_S(A,W,X)=\P(S=1\mid m(A,W,X))$, thereby allowing the selection mechanism to depend on the outcome regression itself. We begin by introducing the notation underlying the proposed collaborative efficient estimator. This set of notation will be used both in the construction of the estimator and in the derivation of its semiparametric properties. Let $ \gamma = (m, e_A, e_S, r_S, f, k, g, t) $ denote the collection of relevant nuisance functions, which are defined as follows and are shown in Table \ref{tab:notation}:

\begin{align*}
    m(a,w,x)&=\E(Y\mid A=a, W=w, X=x, S=1)\\
    e_A(w)&=\P(A=1\mid W = w,S=0)\\
    e_S(w)&=\P(S=1\mid W=w)\\
    r_S(a,w,x)&=\P(S=1\mid m(A, W, X) = m(a, w, x))\\
    \p &=\P(S=0)\\
    k(w)&=\E\{f(W)\mid e_S(W)=e_S(w)\}.
\end{align*}

We start by introducing the following function $\U_\tau(O; \gamma)$, which will characterize the first-order bias of a plug-in estimator for $\tau$: 

{\small
\begin{align*}
   \U_\tau(O;\gamma) = \frac{1}{\p}\bigg[&\frac{S\{A-e_A(W)\}}{e_A(W)\{1-e_A(W)\}}\frac{1-r_S(A, W, X)}{r_S(A, W, X)}\times\{Y-m(A, W, X)\}\\
    &+\frac{(1-S)\{A-e_A(W)\}}{e_A(W)\{1-e_A(W)\}}\times\{m(A,W,X)-h(A,W)\}\\
    &+\{1-e_S(W)\}\times \{f(W) - k(W)\}\\
             &+(1-S)\{k(W)-\tau\}\bigg].
\end{align*}}

The quantity $r_S(A,W,X)$ is central to the efficiency gain of the proposed estimator. It denotes the probability of belonging to the source population conditional on the outcome regression $m(A,W,X)$. The efficiency gain is achieved by conditioning the population indicator $S$ on the summary-level statistic $m(A,W,X)$, rather than on the full set of variables $(A,W,X)$, thereby reducing dimensionality. This intuition is reflected in the second-order bias term of the plug-in estimator, which depends on $m(A,W,X)$ as one of its key components. We have the following theorem showing that the first-order bias of the plug-in estimator for $\tau$ satisfies the von Mises expansion so that $\U_\tau(O; \gamma)$ can play a valid role as the EIF in one-step correction.

\begin{theorem}[First-order bias of plug-in estimator for $\tau$] \label{theorem:tau}
For any estimator $\hat\gamma=(\hat m, \hat e_A, \hat e_S,  \hat r_S, \hat f, \hat k, \hat g, \hat t)$, 
let $\hat\tau$ denote the plug-in collaborative estimator defined as 
\[\hat\tau = \frac{1}{n\times \hat \p}\sum_{i=1}^n\one\{S_i=0\}\hat k(W_i).\] Assume that for some $\delta > 0$, we have $0 < r_S(A,W,X) < 1$ , $0 < e_A(W) < 1$, and the estimators satisfy $\P(\delta<\hat r_S(A,W,X)<1-\delta) = 1$ and
$\P(\delta<\hat e_A(W)<1-\delta) = 1$. The bias of $\hat\tau$ is equal to
$\hat\tau - \tau = -\E[\U_{\tau}(O;\hat \gamma)] + R_\U(\gamma, \hat \gamma),$ where $R_\U$ is a second-order term (defined in the proof of the theorem) equal to a sum of products of errors of the type $\int c(\gamma,\hat\gamma)(a(\hat\gamma) - a(\gamma))(b(\hat\gamma) - b(\gamma))\dd\P$ and is upper bounded as
{\footnotesize
\[
R_\U(\gamma, \hat \gamma) \lesssim 
\| \hat{r}_S - r^*_S \| \cdot \| \hat{m} - m \| 
+ \| \hat{e}_S - e_S \| \cdot \| \hat{k} - k \| 
+ \| \hat{e}_S - e_S \| \cdot \| \hat{f} - f \| 
+ \| \hat{e}_S - e_S \| \cdot \| k^* - k \|\] 
\[
+ \| \hat{e}_A - e_A \| \cdot \| \hat{f} - f \| 
+ \| \hat{e}_A - e_A \| \cdot \| \hat{g} - g \| 
+ (\hat{t} - t)(\hat{\tau} - \tau),
\]}

where we define $r^*_S(a,w,x) = \P(S=1\mid \hat m(A,W,X) = \hat m(a,w,x), m(A,W,X) = m(a,w,x))$ and $k^*(w) =  \E[f(W) \mid e_S(W) = e_S(w), \hat e_S(W) = \hat e_S(w)]$. The proof for this theorem is included in the Appendix \ref{subsec:vonmise_tau}.  
\end{theorem}

The proof of this theorem proceeds in a manner closely analogous to that of Theorem \ref{theorem:lambda}. Inspection of $\U_\tau(O; \gamma)$ and comparison with $\D_\lambda(O;\eta)$ can further explain expected efficiency gains when using $\tau$ vs $\lambda$ for causal effect transportation.  Specifically, the first term of $\U_\tau(O; \gamma)$ involves division by $r_S(A,W,X)$ instead of $q_S(A,W,X)$, which is used in previous estimators. Since $r_S(A,W,X)$ conditions on dimension reduction of $W$ instead of the full vector, it may provide a more stable inverse probability weight that improves the estimator's efficiency. The extent to which this may lead to efficiency gains will be illustrated in Section \ref{sec:simulation} through simulations, showing cases where $\tau$ is more efficient than $\lambda$ or $\theta$. 

\section{Proposed Estimators}
\label{sec:estimator}
\subsection{Doubly robust unbiased transformation for CATE}
Our proposed estimators are built from the estimating functions $\D_\lambda(O; \eta)$ and $\U_\tau(O; \gamma)$. While estimation for $\lambda$ and $\theta$ enjoys double robustness, estimation of $\tau$ does not. For all estimators, Theorems~\ref{theorem:lambda} and~\ref{theorem:tau} show that the estimation error of the target population CATE, $f(W)$, enters most of the second-order remainder terms, so the performance of the proposed estimators depends on the consistent estimate of $f(W)$.  In order to obtain a doubly robust estimator for $f(W)$, we will use a \textit{doubly-robust unbiased pseudo-outcome}, a technique that is regularly used in causal inference problems to estimate high-dimensional parameters \citep{rubin2007doubly, luedtke2016super, kennedy2017non, kennedy2023towards} and that was similarly used in \citet{rudolph2025improving}.\\

Inspection of the definition of $h(a,w)$ reveals that its estimation involves integrating out the mediator and intermediate confounder (namely the vector $X$) from the outcome regression $m(A, W, X)$ with respect to $p_X(a,w) := \P(X = x\mid A = a, W = w, S = 0)$. In what follows, we assume that this integral can be easily computed (e.g., $X$ is discrete). Our strategy to estimate $f$ is based on the following result. 

\begin{lemma}[Doubly robust unbiased transformation for CATE estimation in target population]\label{lemma:DR}
Let $\kappa = (e_A,  e_S, q_S, m, p_X, h)$ denote the collection of relevant nuisance parameters. Assume $0 < e_A(W) < 1$ and $1 - e_S(W) > 0$ almost surely. Define $\D_h(O;\kappa) = \D_{h,1}(O;\kappa) - \D_{h,0}(O;\kappa)$, where
\begin{align*}
    \D_{h, a}(O; \kappa) &= \frac{\one(A = a, S = 1)}{e_A(W)^a\{1-e_A(W)\}^{1-a} \{1 - e_S(W)\}} \times \frac{1-q_S(a,W,X)}{q_S(a,W,X)}[Y - m(a,W,X)] \\
  &\qquad + \frac{\one(A = a, S=0)}{e_A(W)^a\{1-e_A(W)\}^{1-a} \{1 - e_S(W)\}} \times [m(a,W,X) - h(a,W)] \\
  &\qquad + h(a,W).
\end{align*}
For arbitrary $\kappa'$ we have $f(w) = \E[\D_{h}(O; \kappa') \mid W=w]$ and $g(w) = \E[\D_{h,0}(O; \kappa') \mid W=w]$ whenever $\kappa'$ is such that following two conditions hold: (1) $\{q'_S = q_S$ and $e'_S = e_S$ and $e'_A = e_A\}$ or $m' = m$, (2) $\{e'_S = e_S$ and $e'_A = e_A\}$ or $p'_X = p_X$.
\end{lemma}

The above lemma inspires our construction of an estimator $\tilde f$ obtained by regressing a pseudo-outcome $\D_h(O_i;\hat\kappa)$ on covariates $W$, where $\hat\kappa$ denotes a preliminary estimate of the nuisance parameters described in the lemma. Analogous estimators of $g(w)$ may be obtained by using $\D_{h,0}(O_i;\hat\kappa)$ as a pseudo-outcome as $\tilde g$. We provide the proof in Appendix \ref{subsec:DR_CATE}. In what follows, we describe in detail the proposed algorithms to estimate $\theta$, $\lambda$, and $\tau$. 

\subsection{Estimation for $\theta$, $\lambda$,  and $\tau$}

In this section, we propose one-step semiparametric estimators for $\theta$, $\lambda$, and $\tau$. The procedure begins with initial plug-in estimators, denoted by $\hat{\theta}$, $\hat{\lambda}$, and $\hat{\tau}$, which are obtained by substituting estimated nuisance functions into the target parameter definitions. These initial estimates are then corrected by adding the empirical mean of the corresponding influence function. Specifically: $ \tilde{\theta} = \hat{\theta} + \frac{1}{n}\sum_{i = 1}^{n}\D_\theta(O_i;\hat{\xi})$, $ \tilde{\lambda} = \hat{\lambda} + \frac{1}{n}\sum_{i = 1}^{n}\D_\lambda(O_i;\hat{\eta})$, and $ \tilde{\tau} = \hat{\tau} + \frac{1}{n}\sum_{i = 1}^{n}\U_\tau(O_i;\hat{\gamma})$.  Incorporating a doubly robust unbiased transformation for estimating $f(W)$ as described in the previous subsection, we can implement three estimators as follows:

\renewcommand{\labelitemi}{} 
\begin{itemize}
    \item \textbf{Step 1} First, we calculate estimates of $\hat{m}(A,W,X)$, $\hat{e}_A(W)$, $\hat{e}_S(W)$, $\hat{q}_S(A,W,X)$, $\hat{r}_S(A,W,X)$, and $\hat{p}_X(A,W)$ by fitting a data-adaptive regression that incorporates machine learning. We will apply this estimation approach to all subsequent regression estimations.
    
    \item \textbf{Step 2} Compute a doubly robust estimator of $f(w)$ by regressing $\D_h(O_i;\hat\kappa)$ on $W$. Denote the resulting estimator by $\tilde f(w)$.
    \begin{itemize}[label={},leftmargin=*]
        \item \textbf{Step 2.1} For each $a\in\{0,1\}$, compute $\hat h(a,w)$ by integration. For example, if $X$ is discrete, compute
        \[\hat h(a,w) = \sum_{x \in \mathcal{X}} \hat m(a,w,x)\, \hat p_X(a,w)\]
        \item \textbf{Step 2.2} For each observation $i$, compute $\D_{h,1}(O_i;\hat\kappa)$ and $\D_{h,0}(O_i;\hat\kappa)$, and then set
        \[
        \D_h(O_i;\hat\kappa)=\D_{h,1}(O_i;\hat\kappa)-\D_{h,0}(O_i;\hat\kappa),
        \]
        where $\hat\kappa=(\hat e_A,\hat e_S,\hat q_S,\hat m,\hat p_X,\hat h)$.
        
        \item \textbf{Step 2.3} Regress $\D_h(O_i;\hat\kappa)$ on $W$ to obtain $\tilde f(W)$, and regress $\D_{h,0}(O_i;\hat\kappa)$ on $W$ to obtain $\tilde h(0,W)$, denoted by $\tilde g(W)$.
    \end{itemize}
    \item \textbf{Step 3} Compute the doubly robust transformed estimate of $h(A,W)$ as $\tilde{h}(A,W)$ where $ \tilde{h}(A,W) = A \times \tilde{f}(W) + \tilde{g}(W)$.

    \item \textbf{Step 4} Using estimated nuisance parameters, compute $\D_\theta(O_i;\hat\xi)$ for each observation. Estimate
    $\hat\theta=\frac{1}{n\hat \p}\sum_{i=1}^n \one(S_i=0)\tilde f(W_i),$
    and then compute the one-step estimator $\tilde\theta=\hat\theta+\frac{1}{n}\sum_{i=1}^n \D_\theta(O_i;\hat\xi).$

    \item \textbf{Step 5} With known $V$ satisfying Assumption~\ref{ass:modif}, regress $\tilde f(W)$ on $V$ to obtain an estimate $\tilde l(V) = \hat \E\{\tilde f(W)\mid V\}$. Then compute the plug-in estimator
    $\hat\lambda=\frac{1}{n\hat t}\sum_{i=1}^n \one(S_i=0)\tilde l(V_i).$
    Using estimated nuisance parameters, construct $\D_\lambda(O_i;\hat\eta)$ and compute
    $\tilde\lambda=\hat\lambda+\frac{1}{n}\sum_{i=1}^n \D_\lambda(O_i;\hat\eta)$
    where we have $\hat\eta=(\hat m,\hat e_A,\hat e_S,\hat q_S,\tilde f,\tilde l,\tilde g,\hat t)$.

    \item \textbf{Step 6} Compute $\tilde k(W)$ by regressing $\tilde f(W)$ on the estimated selection score $\hat e_S(W)$---$\tilde k(W) = \hat \E\{\tilde f(W) \mid \hat e_S(W)\}$. Then compute the plug-in estimator $\hat\tau=\frac{1}{n\hat \p}\sum_{i=1}^n \one\{S_i=0\}\tilde k(W_i).$ Using estimated nuisance parameters, construct $\U_\tau(O_i;\hat\gamma)$ and compute
   $\tilde\tau=\hat\tau+\frac{1}{n}\sum_{i=1}^n \U_\tau(O_i;\hat\gamma),$
    where we have $\hat\gamma=(\hat m,\hat e_A,\hat e_S,\hat r_S,\tilde f,\tilde k,\tilde g,\hat t)$.

\end{itemize}

Next, we derive asymptotic results to justify the use of Wald-type confidence intervals and hypothesis tests based on each of our proposed estimators. These procedures rely on asymptotic normality and are valid under regularity conditions that ensure appropriate convergence rates for the nuisance parameters. In particular, the following theorems establish the weak convergence of the estimator under certain assumptions.

\begin{theorem}[Weak convergence of $\tilde{\lambda}$]
Assume the second-order term $R_\D(\eta, \hat{\eta}) = o_\P(n^{-1/2})$, and that $\hat{q}_S(A,W,X)$ and $\hat{e}_A(W)$ are bounded away from 0 and 1 with probability 1. In addition, assume that the function $\D_\lambda(O; \hat{\eta})$ is in a Donsker class. Then we have $\sqrt{n}(\tilde{\lambda} - \lambda) \xrightarrow{d} \mathcal{N}\left(0, \operatorname{Var}[\D_\lambda(O; \eta)]\right)$.
\end{theorem}

\begin{theorem}[Weak convergence of $\tilde{\tau}$]
    Assume the second-order term $R_\U(\gamma, \hat{\gamma}) = o_\P(n^{-1/2})$, and that $\hat{r}_S(A,W,X)$ and $\hat{e}_A(W)$ are bounded away from 0 and 1 with probability 1. In addition, assume that the function $\U_\tau(O; \hat{\gamma})$ is in a Donsker class. Then we have $\sqrt{n}(\tilde{\tau} - \tau) \xrightarrow{d} \mathcal{N}\left(0, \operatorname{Var}[\U_\tau(O; \gamma)]\right)$.
\end{theorem}

The proof of the above two theorems follows the same steps given in \citet{kennedy2024semiparametric}. The Donsker assumption can be satisfied whenever the nuisance estimators are in classes with bounded entropy integrals such as Vapnik-Chervonenkis classes. If this assumption is too restrictive given the desired nuisance estimators, one can instead rely on sample-splitting procedures that do away with this assumption \citep{zheng2011cross, chernozhukov2018double, kennedy2024semiparametric}. Since $R_\D$ and $R_\U$ are second-order terms, the condition of the order of convergence being $o_\P(n^{-1/2})$ would hold if each component of $\eta$ and $\gamma$ is estimated at a rate $ o_\P(n^{-1/4})$ or faster. In addition, as can be seen in the proof of the theorems, control of these remainder terms also requires  \( (\hat{t} - t)(\hat{\lambda} - \lambda) = o_\P(n^{-1/2}) \) and \( (\hat{t} - t)(\hat{\tau} - \tau) = o_\P(n^{-1/2}) \). This holds as long as the plug-in estimators  $\hat{\lambda}$ and $\hat{\tau}$ are consistent, since $\hat{t} - t = O_\P(n^{-1/2})$ by the central limit theorem.

\begin{coro}
    Assume the conditions in Theorem \ref{theorem:lambda} hold and that $R_\D(\eta,\hat\eta) = o_\P(1)$, 
    then $\tilde{\lambda}$ is consistent.
\end{coro}
A sufficient condition for $R_\D(\eta,\hat\eta) = o_\P(1)$ is that at least one of the following holds:
{\scriptsize
\[
\text{(i)}\quad 
\{\|\hat q_S - q_S\| = o_\P(1)\ \text{or}\ \|\hat m - m\| = o_\P(1)\}
\, \text{and} \, \|\hat e_S - e_S\| = o_\P(1)
\, \text{and} \, \|\hat e_A - e_A\| = o_\P(1)
\]

\[
\text{(ii)}\quad 
\{\|\hat q_S - q_S\| = o_\P(1)\ \text{or}\ \|\hat m - m\| = o_\P(1)\}
\, \text{and} \, \|\hat e_S - e_S\| = o_\P(1)
\, \text{and} \, \|\hat f - f\| = o_\P(1) \,
\text{and} \, \|\hat g - g\| = o_\P(1)
\]

\[
\text{(iii)}\quad 
\{\|\hat q_S - q_S\| = o_\P(1)\ \text{or}\ \|\hat m - m\| = o_\P(1)\}
\, \text{and} \, \|\hat e_A - e_A\| = o_\P(1)
\, \text{and} \, \|\hat f - f\| = o_\P(1)
\, \text{and} \, \|\hat l - l\| = o_\P(1) \
\]

\[
\text{(iv)}\quad 
\{\|\hat q_S - q_S\| = o_\P(1)\ \text{or}\ \|\hat m - m\| = o_\P(1)\}
\, \text{and} \, \|\hat l - l\| = o_\P(1)
\, \text{and} \, \|\hat f - f\| = o_\P(1)
\, \text{and} \, \|\hat g - g\| = o_\P(1).
\]
}

The last term in $R_\D(\eta,\hat\eta)$ involves $\hat{t} - t$, where $\hat{t}$ is the empirical average of $S = 0$. By the weak law of large numbers, $\hat{t} - t = o_\P(1)$. All conditions require either a consistent estimation of the outcome regression $m$ or a consistent estimation of $q_S$.

\begin{coro}
    Assume the conditions in Theorem \ref{theorem:tau} hold and that $R_\U(\gamma, \hat{\gamma}) = o_\P(1)$, 
    then $\tilde{\tau}$ is consistent.
\end{coro}

A sufficient condition for $R_\U(\gamma, \hat{\gamma}) = o_\P(1)$ is that at least one of the following holds:
{\scriptsize
\[
\text{(i)}\quad 
\{\|\hat r_S - r_S^*\| = o_\P(1)\ \text{or}\ \|\hat m - m\| = o_\P(1)\}
\, \text{and} \, \|\hat e_S - e_S\| = o_\P(1)
\, \text{and} \, \|\hat e_A - e_A\| = o_\P(1)
\]

\[
\text{(ii)}\quad 
\{\|\hat r_S - r_S^*\| = o_\P(1)\ \text{or}\ \|\hat m - m\| = o_\P(1)\}
\, \text{and} \, \|\hat e_S - e_S\| = o_\P(1)
\, \text{and} \, \|\hat f - f\| = o_\P(1) \,
\text{and} \, \|\hat g - g\| = o_\P(1)
\]

\[
\text{(iii)}\quad 
\{\|\hat r_S - r_S^*\| = o_\P(1)\ \text{or}\ \|\hat m - m\| = o_\P(1)\}
\, \text{and} \, \|\hat k - k\| = o_\P(1)
\, \text{and} \, \|\hat f - f\| = o_\P(1)
\, \text{and} \, \|\hat {e}_A - e_A\| = o_\P(1) \, 
\text{and} \, \|k^* - k\| = o_\P(1)
\]

\[
\text{(iv)}\quad 
\{\|\hat r_S - r_S^*\| = o_\P(1)\ \text{or}\ \|\hat m - m\| = o_\P(1)\}
\, \text{and} \, \|\hat k - k\| = o_\P(1)
\, \text{and} \, \|\hat f - f\| = o_\P(1)
\, \text{and} \, \|k^* - k\| = o_\P(1)
\, \text{and} \, \|\hat g - g\| = o_\P(1).
\]
}

Finally, note that the last term in $R_\U(\gamma, \hat{\gamma})$ involves $\hat{t} - t$, where $\hat{t}$ is simply the empirical average of $S = 0$. By the weak law of large numbers,  we have that $\hat{t} - t = o_\P(1)$. Observing the above conditions for $R_\U(\gamma, \hat{\gamma}) = o_\P(1)$, it is important to note that all conditions require either consistent estimation of the outcome regression $m$ or require $\hat{r}_S$ to converge to $r^*_S$ in probability.

\subsection{Additional considerations on estimation for $\tau$}

 The collaborative estimator $\tilde \tau$ builds on methods originally developed for the non-transported ATE by \citet{benkeser2020nonparametric}. Its key feature is the collaborative estimation of the selection score and the outcome regression, designed to mitigate the challenges posed by the strong positivity assumption and practical positivity violations \citep{gruber2010application, stitelman2010collaborative, wang2011finding}. By estimating the selection score as a function of the outcome model, $\tilde \tau$ can achieve super-efficiency---attaining smaller asymptotic variance than regular efficient estimators---while also improving robustness to practical positivity violations and maintaining theoretical guarantees for its asymptotic distribution. However, $\tilde \tau$ suffers from regularity issues: it is an irregular estimator, and its nominal confidence intervals may undercover the true estimand, particularly in the presence of positivity violations and smaller sample sizes \citep{benkeser2020nonparametric}. Moreover, the practical implications of its non-regularity remain insufficiently understood, as noted by \citet{rudolph2025improving}.

\section{Simulation studies}
\label{sec:simulation}
In this section, we present the results of a numerical study conducted with the following objectives: 1) to ensure correct implementation of our proposed estimators, 2) to illustrate under which data-generating mechanisms (DGMs) the proposed estimators achieve observable efficiency gains, and 3) to illustrate and compare the finite-sample performance across sample sizes, DGMs, and estimators. We compare the three one-step estimators introduced in previous sections: $\tilde\theta$, which builds on the estimator of \citet{rudolph2017robust}, and our two proposed estimators $\tilde\lambda$ and $\tilde\tau$. We consider three DGMs as shown in Table \ref{tab:DGMs_site} to show no efficiency gain, efficiency gain for $\lambda$, and efficiency gain for $\tau$, respectively. \\

DGM 1 is constructed under a scenario where we do not expect an efficiency gain of our proposed estimators over $\tilde \theta$. To construct this scenario in DGM 1, we have two covariates $W_1$ and $W_2$ and specify $V$ as $W_2$ while letting $W_1$ depend only on $W_2$. Therefore, we have $S \indep W_1 \mid  V$ when $V = W_2$ so that \ref{ass:modif} is satisfied. We use this DGM to demonstrate our discussion made after Lemma \ref{lemma:gains}, where even though \ref{ass:modif} is satisfied, every downstream variable carries through effect through $W_2$. Hence taking $V = W_2$ the component $f(W) - \E[f(W) \mid V] $ can be close to zero where no efficiency gain is expected. \\

DGM 2 is constructed to demonstrate the efficiency gain for $\lambda$ over $\theta$. Having \ref{ass:modif} holds in this DGM, we ensure that $V = W_2$ is not predictive of $f(W)$ and thus $[f(W) - \E\{f(W)|V\}]$ is large, as discussed in  Lemma \ref{lemma:gains}. DGM 2 introduces interactions between $W_1$ and $A$ in $Y$, creating a more complex relationship among variables other than $V$ so that $V$ by itself can be less predictive of the CATE. An efficiency gain for $\lambda$ is expected in the simulation of finite samples under DGM 2, as we expect a greater difference for the EIF variances compared to $\theta$.  \\
 
DGM 3 is designed to favor $\tilde \tau$ in terms of efficiency. Although $W_1$ and $W_2$ strongly distinguish the source and target populations, they are weak predictors of the outcome. As a result, the full selection mechanism $\P(S=1\mid A,W,X)$ may be highly dispersed, leading to unstable estimation. In contrast, the collaborative estimator $\tilde \tau$ uses the reduced selection mechanism $\P(S=1\mid m(A,W,X))$. This dimension reduction removes variation in the selection mechanism unrelated to the outcome while retaining information carried by outcome-relevant variables such as $M$, $Z$, $W_3$, and $W_4$. Because the conditional variance of $Y$ is also larger when $M=1$, smoothing the selection probabilities while retaining information through $M$ can yield more stable weights and lower variance for the estimator of $\tilde \tau$. \\

For each DGM, we generated 1000 datasets of sizes $n \in \{500, 1000, 5000, 10000\}$. All estimations of nuisance parameters are obtained using the Super Learner algorithm \citep{vanderLaanPolleyHubbard07}, including a simple mean predictor (SL.mean), standard generalized linear models (SL.glm), penalized generalized linear models (SL.glmnet) \citep{tibshirani1996regression}, multivariate adaptive regression splines (SL.earth) \citep{earth}, and generalized additive models (SL.gam) \citep{Hastie&Tibshirani90}. The Super Learner algorithm uses cross-validation to evaluate the performance of each model and then estimates the optimal set of weights to combine their predictions. We evaluate the performance of our estimators in terms of absolute bias, absolute bias scaled by $\sqrt{n}$ ($\sqrt{n}$-bias), mean squared error scaled by sample size ($n$MSE), efficiency ratio relative to $\tilde \theta$, and coverage of the 95\% confidence interval. The results of the simulation studies are presented in Table~\ref{tab:sim_results_onestep}.\\

For DGM 1, all three estimators exhibit small absolute and $\sqrt{n}$-scaled biases across sample sizes, with empirical coverage remaining between 94\% and 96\%. Because the reported MSE is multiplied by $n$, the approximately constant $n$MSE values are consistent with the expected $n^{-1}$ convergence rate. Three estimators perform similarly in terms of efficiency across all sample sizes. \\

For DGM 2, the absolute biases remain small, and the $\sqrt{n}$-scaled biases are also modest across sample sizes. Empirical coverage ranges from 89\% to 96\%, where slightly lower coverage from $\tilde \tau$ can happen due to irregularity and unstable variance estimation \citep{benkeser2020nonparametric}. Potential solutions are discussed recently for improved sampling variance estimation \citep{ji2025targeted}. The estimator $\tilde \lambda$ consistently has 31\%--36\%  lower $n$MSE than $\tilde \theta$. These findings are consistent with the construction of DGM 2, thereby favoring the estimator $ \tilde \lambda$. The estimator $\tilde \tau$ also improves upon the efficiency of $\tilde \theta$, although its efficiency gains are smaller, with reductions in $n$MSE of approximately 1\%--7\%. \\

For DGM 3, all estimators exhibit small absolute bias and empirical coverage close to 95\% across sample sizes. The $n$MSE values remain relatively stable as $n$ increases. In this setting, $\tilde \theta$ and $\tilde \lambda$ have nearly identical efficiency, whereas $\tilde \tau$ shows a substantial and persistent efficiency advantage, with a 33\%--40\% improvement in $n$MSE relative to $\tilde \theta$. This result is consistent with the design of DGM 3, which was constructed to favor $ \tilde \tau$.

\begin{table}[ht]
\centering
\scriptsize
\begin{tabular}{ll}
\hline
\textbf{Variable} & \textbf{Distribution} \\
\hline
\multicolumn{2}{l}{\textbf{DGM 1}} \\
$S$     & $S \sim \text{Bernoulli}(0.58)$ \\
$W_2$   & $W_2 \sim \text{Bernoulli}(\text{logit}^{-1}(2S - 1))$ \\
$W_1$   & $W_1 \sim \text{Bernoulli}(\text{logit}^{-1}(2W_2 + 0.7))$ \\
$A$     & $A \sim \text{Bernoulli}(\text{logit}^{-1}(0.2S + 0.2W_1 + 0.2W_2 - 0.3))$ \\
$Z$     & $Z \sim \text{Bernoulli}(\text{logit}^{-1}(0.3A + 0.4W_1 + 0.3W_2 - 0.7))$ \\
$M$     & $M \sim \text{Bernoulli}(\text{logit}^{-1}(0.3A + 0.6W_1 + 0.25W_2 - 0.55S + \cos(Z) - 2))$ \\
$Y$     & $Y \sim  \mathcal{N}(0.4W_1 - 0.3A + 7M + 0.33Z + 0.92W_2,\, 0.5^2)$ \\
True ATE on $S = 1$ & 0.097 \\
True ATE on $S = 0$ & 0.145 \\
\hline
\multicolumn{2}{l}{\textbf{DGM 2}} \\
$S$     & $S \sim \text{Bernoulli}(0.5)$ \\
$W_2$   & $W_2 \sim \text{Bernoulli}(\text{logit}^{-1}(S - 0.5))$ \\
$W_1$   & $W_1 \sim \mathcal{N}(0,1)$ \\
$A$     & $A \sim \text{Bernoulli}(\text{logit}^{-1}(0.2S + 0.2W_1 + 0.3W_2))$ \\
$Z$     & $Z \sim \text{Bernoulli}(\text{logit}^{-1}(0.4A + 0.3W_1 - 0.2))$ \\
$M$     & $M \sim \text{Bernoulli}(\text{logit}^{-1}\left(0.3Z + 0.2W_1 + (6 - 2S)A - 3\right))$ \\
$Y$     & $Y \sim \mathcal{N}(0.5W_1 + 0.2W_2 + 0.5Z + 3M + A(3W_1 - 1) + 1,\,0.5^2)$ \\
True ATE on $S = 1$  & 1.156 \\
True ATE on $S = 0$ & 1.757 \\
\hline
\multicolumn{2}{l}{\textbf{DGM 3}} \\
$S$     & $S \sim \text{Bernoulli}(0.5)$ \\
$W_1$   & $W_1 \sim \text{Bernoulli}(\text{logit}^{-1}(1.4S - 0.2))$ \\
$W_2$   & $W_2 \sim \text{Bernoulli}(\text{logit}^{-1}(1.2S - 0.3))$ \\
$W_3$   & $W_3 \sim \text{Bernoulli}(\text{logit}^{-1}(0.10S - 0.1))$ \\
$W_4$   & $W_4 \sim \text{Bernoulli}(\text{logit}^{-1}(0.10S - 0.2))$ \\
$A$     & $A \sim \text{Bernoulli}(\text{logit}^{-1}(0.15S + 0.25W_3 + 0.20W_4 - 0.3))$ \\
$Z$     & $Z \sim \text{Bernoulli}(\text{logit}^{-1}(0.6A + 0.25W_3 + 0.20W_4 - 0.5))$ \\
$M$     & $M \sim \text{Bernoulli}(\text{logit}^{-1}(-0.8 + 0.4A + 0.25Z - 0.4S + 0.45W_3 + 0.45W_4 + 2A(1-S)))$ \\
$Y$     & $Y \sim  \mathcal{N}(0.8A + 2M + 0.4Z + 0.3W_3 + 0.25W_4,\, (0.5 + 0.6M)^2)$\\
True ATE on $S = 1$ & 1.062\\
True ATE on $S = 0$ & 1.777 \\
\hline
\end{tabular}
\caption{Data Generating Mechanisms with Site Heterogeneity}
\label{tab:DGMs_site}
\end{table}

\begin{table}[!htbp]
\scriptsize
\centering
\begin{tabular}{cc|ccc|ccc|ccc}
\toprule
\textbf{$n$} & \textbf{Metric} & \multicolumn{3}{c|}{\textbf{DGM 1}} & \multicolumn{3}{c|}{\textbf{DGM 2}} & \multicolumn{3}{c}{\textbf{DGM 3}} \\
\midrule
& & $\tilde\theta$ & $\tilde\lambda$ & $\tilde\tau$
& $\tilde\theta$ & $\boldsymbol{\tilde\lambda}$ & $\tilde\tau$
& $\tilde\theta$ & $\tilde\lambda$ & $\boldsymbol{\tilde\tau}$ \\
\midrule
\multirow{5}{*}{500} 
& $|\text{Bias}|$ 
& 0.006 & 0.006 & 0.007 
& 0.006 & 0.001 & 0.004 
& 0.000 & 0.004 & 0.005 \\
& $|\sqrt{n}\text{-Bias}|$ 
& 0.141 & 0.132 & 0.156 
& 0.123 & 0.013 & 0.082 
& 0.009 & 0.078 & 0.114 \\
& $n$MSE 
& 113 & 114 & 113 
& 25.7 & 16.8 & 25.6 
& 27.3 & 27.1 & 18.3 \\
& Efficiency ratio vs.\ $\theta$ 
& 1.000 & 1.010 & 1.000 
& 1.000 & 0.653 & 0.996 
& 1.000 & 0.992 & 0.669 \\
& 95\% C.I. Coverage 
& 94\% & 94\% & 95\% 
& 94\% & 94\% & 90\% 
& 93\% & 93\% & 92\% \\
\midrule
\multirow{5}{*}{1000} 
& $|\text{Bias}|$ 
& 0.016 & 0.015 & 0.016 
& 0.004 & 0.001 & 0.006 
& 0.002 & 0.005 & 0.003 \\
& $|\sqrt{n}\text{-Bias}|$ 
& 0.495 & 0.477 & 0.506 
& 0.133 & 0.030 & 0.190 
& 0.070 & 0.172 & 0.085 \\
& $n$MSE 
& 113 & 113 & 113 
& 26.8 & 18.4 & 26.5 
& 26.9 & 27.0 & 17.2 \\
& Efficiency ratio vs.\ $\theta$ 
& 1.000 & 1.000 & 0.997 
& 1.000 & 0.688 & 0.989 
& 1.000 & 1.000 & 0.640 \\
& 95\% C.I. Coverage 
& 95\% & 94\% & 94\% 
& 94\% & 94\% & 89\% 
& 95\% & 95\% & 95\% \\
\midrule
\multirow{5}{*}{5000} 
& $|\text{Bias}|$ 
& 0.003 & 0.002 & 0.002 
& 0.001 & 0.001 & 0.001 
& 0.001 & 0.004 & 0.003 \\
& $|\sqrt{n}\text{-Bias}|$ 
& 0.179 & 0.175 & 0.168 
& 0.086 & 0.103 & 0.077 
& 0.082 & 0.303 & 0.184 \\
& $n$MSE 
& 112 & 112 & 112 
& 25.2 & 16.2 & 24.7 
& 28.9 & 28.8 & 17.5 \\
& Efficiency ratio vs.\ $\theta$ 
& 1.000 & 1.000 & 0.998 
& 1.000 & 0.644 & 0.983 
& 1.000 & 0.998 & 0.605 \\
& 95\% C.I. Coverage 
& 95\% & 95\% & 95\% 
& 95\% & 95\% & 91\% 
& 95\% & 95\% & 95\% \\
\midrule
\multirow{5}{*}{10000} 
& $|\text{Bias}|$ 
& 0.001 & 0.001 & 0.001 
& 0.002 & 0.001 & 0.002 
& 0.001 & 0.004 & 0.001 \\
& $|\sqrt{n}\text{-Bias}|$ 
& 0.093 & 0.098 & 0.068 
& 0.162 & 0.113 & 0.165 
& 0.090 & 0.390 & 0.131 \\
& $n$MSE 
& 110 & 110 & 110 
& 21.9 & 14.7 & 20.3 
& 30.6 & 30.6 & 18.8 \\
& Efficiency ratio vs.\ $\theta$ 
& 1.000 & 0.999 & 0.999 
& 1.000 & 0.671 & 0.928 
& 1.000 & 0.999 & 0.617 \\
& 95\% C.I. Coverage 
& 96\% & 96\% & 95\% 
& 96\% & 96\% & 93\% 
& 95\% & 95\% & 94\% \\
\bottomrule
\end{tabular}
\caption{Simulation results for DGMs using the one-step estimator. Efficiency ratios are computed as $n\mathrm{MSE}(\cdot)/n\mathrm{MSE}(\tilde \theta)$, thus values below 1 indicate greater efficiency than $\tilde \theta$.}
\label{tab:sim_results_onestep}
\end{table}

\section{Application}
\label{sec:application}
\subsection{Moving to Opportunity Study (MTO)}
\subsubsection{Motivations}
We apply the estimators proposed herein to a subset of data from the Moving to Opportunity Study \citep{kling2007experimental,sanbonmatsu2011moving}.
MTO was a randomized trial where families living in public housing developments in five US cities could volunteer to be randomized to either receive a Section 8 housing voucher or not. Section 8 housing vouchers allow low-income families to have their rent subsidized to be no more than 30-40\% of their income on the private market \citep{hud} and this form of housing assistance continues to be the most common form of housing assistance today, with over 2.3 million families receiving this type of assistance in 2022 \citep{CBPP_HCV_web}. After randomization, families were followed for 10-15 years \citep{sanbonmatsu2011moving}. \\

We were interested in transporting the effect of randomized assignment to receive a Section 8 voucher or not on the risk of developing a mood disorder at the 10-15 year follow-up among participants' sons who identified as Black race or Latine ethnicity and who were 0-5 years old at randomization (and therefore adolescents at the 10-15 year follow-up point). We considered three of the MTO sites; we transported from the New York City (NYC) and Chicago sites to the Los Angeles (LA) site. We chose this particular subset, because the sites exhibited qualitative differences in effect estimates; in LA, the estimated additive effect of voucher receipt on risk of mood disorder at follow-up (the risk difference) was 0.056 (95\% CI: -0.011, 0.123; i.e., the voucher increased risk of mood disorder); in NYC and Chicago, the estimated additive effect of voucher receipt on risk of mood disorder at follow-up was -0.052 (95\% CI: -0.104, -0.001; i.e., the voucher decreased risk of mood disorder). When doing this transportation, we do not use outcome data from the LA site. However, in reality, we have this data, so we use the observed estimate from LA as a benchmarking exercise. 

\subsubsection{Measures}
Baseline covariates included race/ethnicity (categorized as Black or Latine), age of the child (at baseline), whether or not the child was reported to have behavioral problems at baseline or be enrolled in a gifted class at baseline, whether or not the parent: graduated high school, was ever married, was working, had a child when they were under age 18,
 was very dissatisfied with the baseline neighborhood, 
 and whether or not the household: received Aid to Families with Dependent Children (a federal welfare program at the time), had applied for Section 8 previously, moved 3 or more times in the previous 5 years, had a member with a disability, and size of household (2, 3, 4+). All baseline covariates had zero missingness, except for race/ethnicity, which had 2\% missingness. For the $\lambda$ parameter that relies on specifying $V$, we included: $V$ = \{Black race; whether or not the parent: graduated high school, was ever married; and whether or not the household: received Aid to Families with Dependent Children, moved 3 or more times in the previous 5 years, or had a member with a disability\}, which previous research suggested may explain some heterogeneity across sites \citep{rudolph2022ends,osypuk2024does,rudolph2025improving}.

The exposure was a randomized assignment to receive a Section 8 housing voucher. An initial important post-exposure variable that differed between the sites and also would affect the outcome was moving with the voucher or not. In Figure 1, this variable would be called $M$, and in Figure 2, this variable would be called $Z$. Assuming the setting depicted in Figure 2, we added a subsequent post-exposure variable related to school poverty over the duration of follow-up: whether or not the child had ever attended a non-Title I school (schools receive a Title I designation if they are considered ``high-poverty'',  meaning that a substantial number of students live below the poverty line). The outcome was presence or absence of a mood disorder (depression, bipolar (including subthreshold), and mania (including subthreshold), which was measured with World Health Organization Composite International Diagnostic Interview (CIDI) Version 3.0, which has been shown to align with Diagnostic Statistical Manual-IV (DSM-IV) diagnoses \citep{kessler2009national}. 

\subsubsection{Estimation} We used the proposed one-step estimators to estimate $\theta$, $\lambda$, and $\tau$ following Figure \ref{fig:dag} and Figure \ref{fig:dag_ic}. For the sake of simplicity in this illustrative example, we used data from a single imputed dataset and used 10-fold cross-fitting. In a real applied analysis, we would use multiple imputed datasets \citep{buuren2010mice}. We truncated the weights used in the EIF at 100.  We used the following algorithms in our stacked ensemble learner (i.e., Super Learner \citep{vanderLaanPolleyHubbard07}): an intercept-only model, a main-terms generalized linear model, LASSO \citep{tibshirani1996regression}, gradient boosted machines \citep{chen2015xgboost}, and multivariate adaptive regression splines \citep{earth}.

\subsubsection{Results}

Figure \ref{fig:mto_scm} shows the transported effect estimates ($\tilde\theta$, $\tilde\lambda$, and $\tilde\tau$) alongside the observed effect estimates in the NYC/Chicago (source) and LA (target) sites, separately for the two data generating mechanisms represented by Figure \ref{fig:dag} (panel a) and Figure \ref{fig:dag_ic} (panel b). Across both estimation settings, $\tilde\tau$ had a visibly narrower confidence interval than $\tilde\theta$ and $\tilde\lambda$, while $\tilde\lambda$ had a similar confidence interval to $\tilde\theta$, which corresponds to the third scenario of the data-generating process from the previous simulation. Note that all point estimates moved toward the direction of the effect observed in the target (LA) population. While there was a difference in point estimates between $\tilde\tau$ and the other two estimators, their confidence intervals overlapped substantially.

\begin{figure}[!htbp]
    \centering
    \captionsetup{width=\linewidth}

    \begin{subfigure}[b]{0.6\textwidth}
        \centering
        \includegraphics[width=\textwidth]{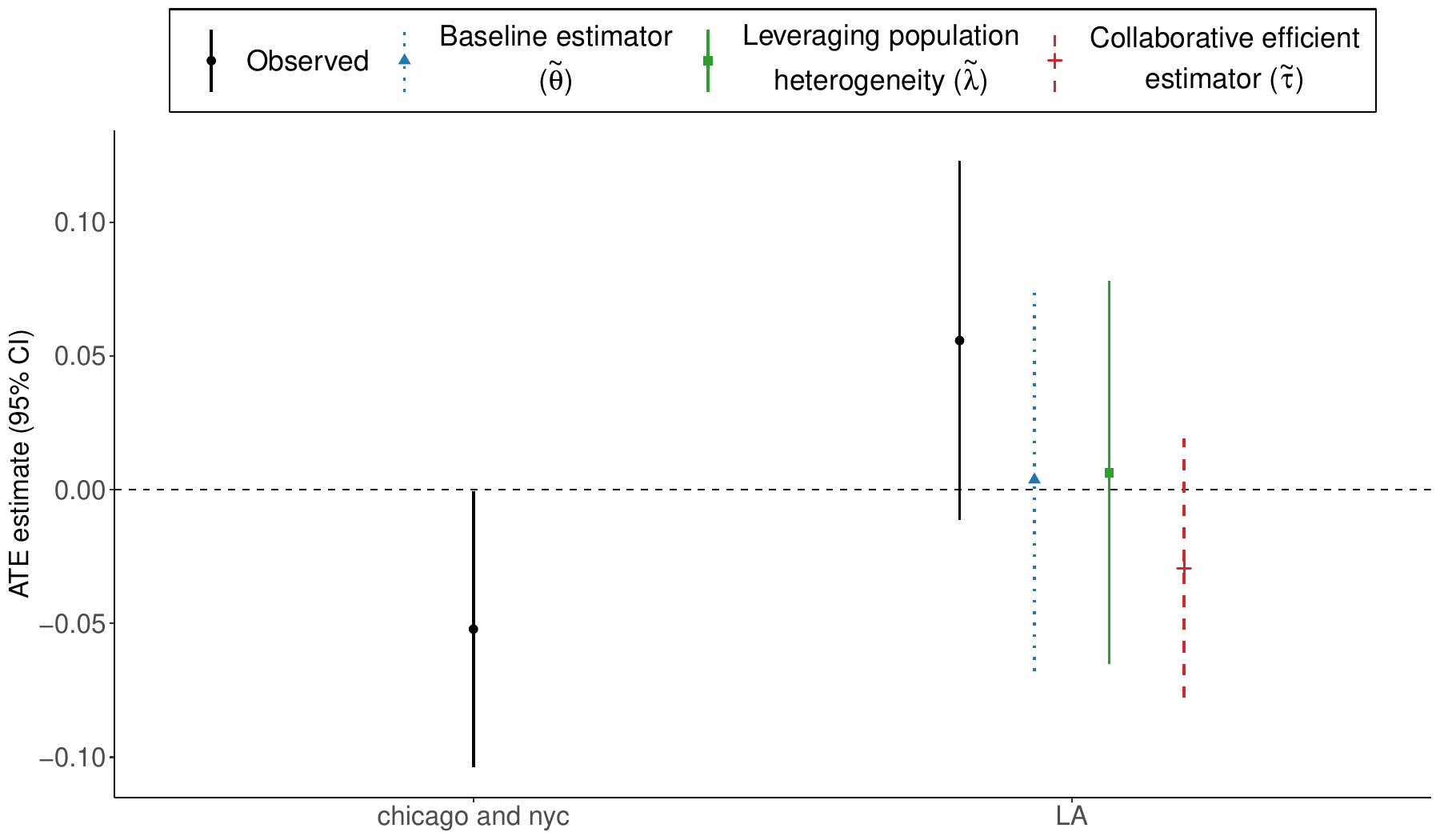}
        \caption{}
        \label{fig:mto_scm1}
    \end{subfigure}

    \vspace{0.3cm}

    \begin{subfigure}[b]{0.6\textwidth}
        \centering
        \includegraphics[width=\textwidth]{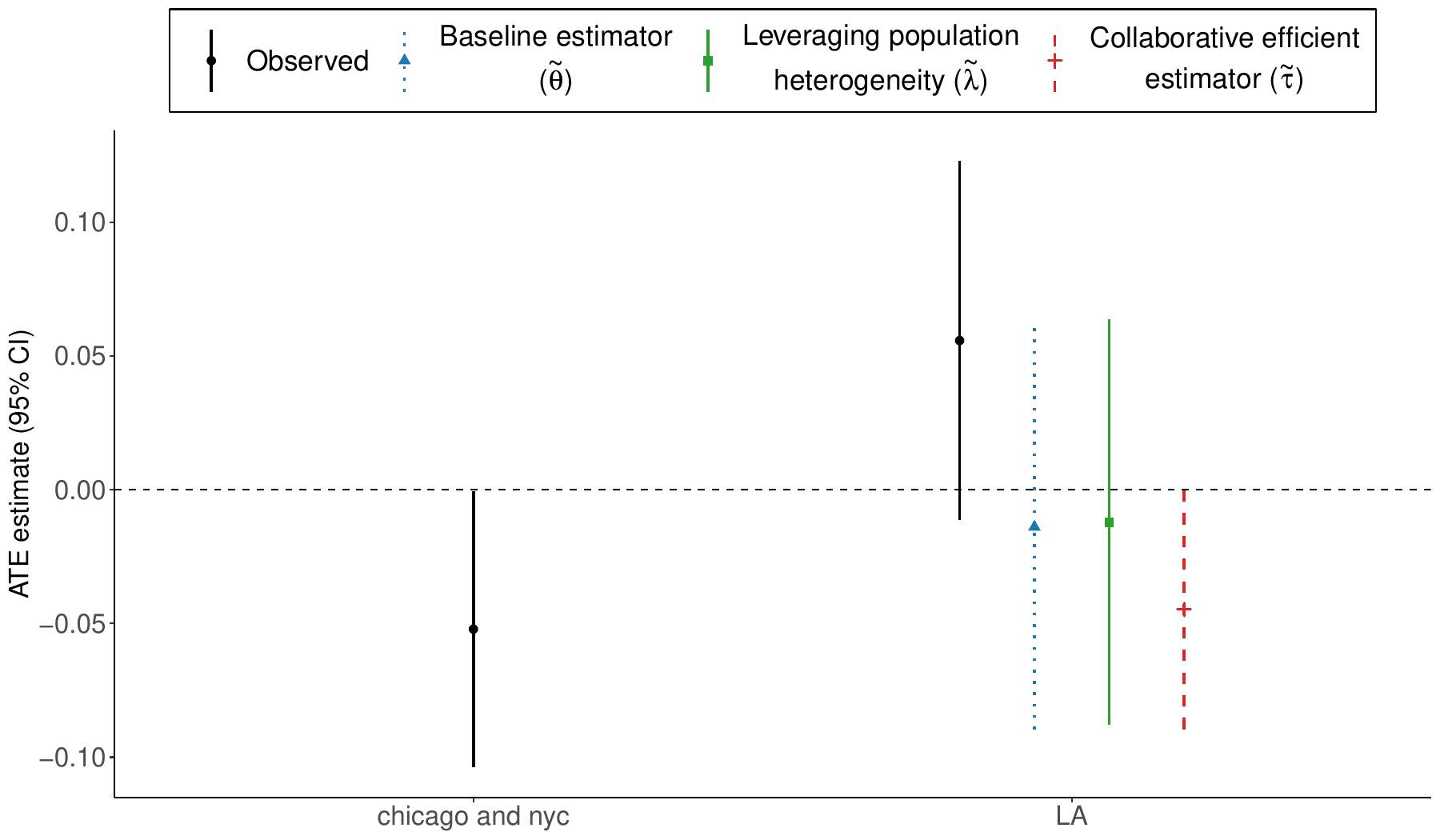}
        \caption{}
        \label{fig:mto_scm2}
    \end{subfigure}

    \caption{Effect estimates and 95\% confidence intervals using data from the MTO study. The ATE is interpreted as the effect of randomized assignment to receive a Section 8 housing voucher, versus not, on the risk of mood disorder among boys. The NYC and Chicago sites are considered the source population, and the LA site is considered the target population. We show results assuming data generated according to Figure \ref{fig:dag} (panel a) and Figure \ref{fig:dag_ic} (panel b). All results were approved for release by the U.S. Census Bureau, authorization number CBDRB-FY26-0336.}
    \label{fig:mto_scm}
\end{figure}

\subsection{The Aging and Cognitive Health Evaluation in Elders (ACHIEVE)}
\subsubsection{Motivations}
The Aging and Cognitive Health Evaluation in Elders (ACHIEVE) study is a multicenter randomized controlled trial of older adults with hearing loss, designed to evaluate the effect of a best-practice hearing intervention compared with a successful aging health education control on cognition change over a 3-year follow-up period (ClinicalTrials.gov Identifier: NCT03243422) \citep{lin2023hearing}. The trial enrolled participants from two distinct cohorts: a de novo community-based recruitment group and an existing cohort from the Atherosclerosis Risk in Communities (ARIC) study \citep{aric1989atherosclerosis}. Although the trial achieved all its recruitment and follow-up targets, the primary outcome results were null. However, a prespecified sensitivity analysis found that the hearing intervention reduced the rate of cognitive decline by 48\% among participants recruited from ARIC. ACHIEVE has demonstrated evidence of a treatment effect within the ARIC subgroup, with an estimated average treatment effect using one step estimator on global cognition score at year 3 of 0.168 (95\% CI: 0.005, 0.333; i.e. hearing intervention is more effective than the health education on participants' cognition). However, there was a null effect within the de novo group with an estimated average treatment effect using one step estimator of -0.019 (95\% CI: -0.096, 0.058; i.e. hearing intervention shows no beneficial effects over health education on participants' cognition). By transporting from the ARIC subgroup population to de novo subgroup population, we are interested in whether the transported estimate can be similar to the observed estimate and consistent with a null effects in de novo population. When doing this transportation, we do not use outcome data from the de novo group. This analysis would also provide evidence on the generalizability of observed hearing aid benefits among healthier population as in de novo group. However, we use the observed estimate from the de novo population as a benchmarking exercise. 

\subsubsection{Measures}
Baseline covariates included demographic characteristics (age, sex, race, education, former occupation, annual income); genetic characteristics (presence of APOE $\varepsilon4$ alleles); hearing-related measures (pure-tone average in the better-hearing ear, Hearing Handicap Inventory for the Elderly score, lifetime loud-noise exposure, and Quick Speech in Noise test score); anthropometric measures (body weight, height, body mass index, waist circumference, hip circumference, and waist-to-hip ratio); cardiovascular measures (systolic and diastolic blood pressure and resting heart rate); self-reported medical conditions (hypertension, diabetes, coronary heart disease, myocardial infarction, stroke, Parkinson’s disease, traumatic brain injury, and seizures); lifestyle factors (cigarette and alcohol use and leisure-time and sport-related physical activity); physical function measures (balance, repeated chair stands, 4-meter walking speed, and grip strength); functional status measures (functional limitations, activities of daily living, and instrumental activities of daily living); mental health (depressive symptoms); and cognitive measures, including the Mini-Mental State Examination, a 10-test cognitive battery, and factor scores derived from the 10-test battery. Details on these measures and corresponding references are included in the Appendix \ref{subsec:ACHIEVE variables}. Specific variables used to leverage heterogeneity across two populations include $V = \{$age, sex, race, education, annual income, systolic blood pressure, waist-to-hip circumference ratio, diabetes, alcohol use, summary score from balance tests, 4-meter walk, sport-related physical activity, Hearing Handicap Inventory for the Elderly, and the language domain cognitive factor score\}, which was suggested in previous analyses \citep{frangakis2009calibration, pike2025cognitive}. \\

Participants were randomized to receive a hearing intervention or the health education control. The primary outcome is a factor score of global cognition measured at year 3 \citep{lin2023hearing}. We let $M$ be the 2-year embedded social network measure assessed during follow-up visits. The social isolation factor has been demonstrated to have a potential mediating role on the causal pathway between the hearing intervention and cognitive impairment \citep{reed2025hearing}. We define a binary mediator by dichotomizing the embedded network score at the sample median. In addition, we let $Z$ be the 6-month Hearing Handicap Inventory for the Elderly—Screening (HHIE-S) score \citep{ventry1983identification, lichtenstein1988diagnostic}. Hearing handicap plays a crucial role in impacting a person’s social connection and communication in group and community settings, which in turn impacts their social embeddedness \citep{weinstein2016relating, jones2019walk, hoang2022interventions, reed2025hearing}. We define a binary indicator that contrasts no hearing handicap with some degree of hearing handicap, consisting of people with mild–moderate or severe hearing handicap. Missing data were addressed using multiple imputation \citep{demirtas2018flexible} following the order of our structural causal model in Figure~\ref{fig:dag_ic}.

\subsubsection{Estimation} We used the proposed one-step estimators to estimate $\theta$, $\lambda$, and $\tau$ following Figure \ref{fig:dag_ic} through 10-fold cross-fitting. We used the following algorithms in our stacked ensemble learner (i.e., Super Learner \citep{vanderLaanPolleyHubbard07}): an intercept-only model, a main-terms generalized linear model, LASSO \citep{tibshirani1996regression}, gradient-boosted machines \citep{chen2015xgboost}, and multivariate adaptive regression splines \citep{earth}.

\subsubsection{Results}

\begin{figure}[!htbp]
    \centering
    \captionsetup{width=\linewidth}
    \includegraphics[width=0.6\textwidth]{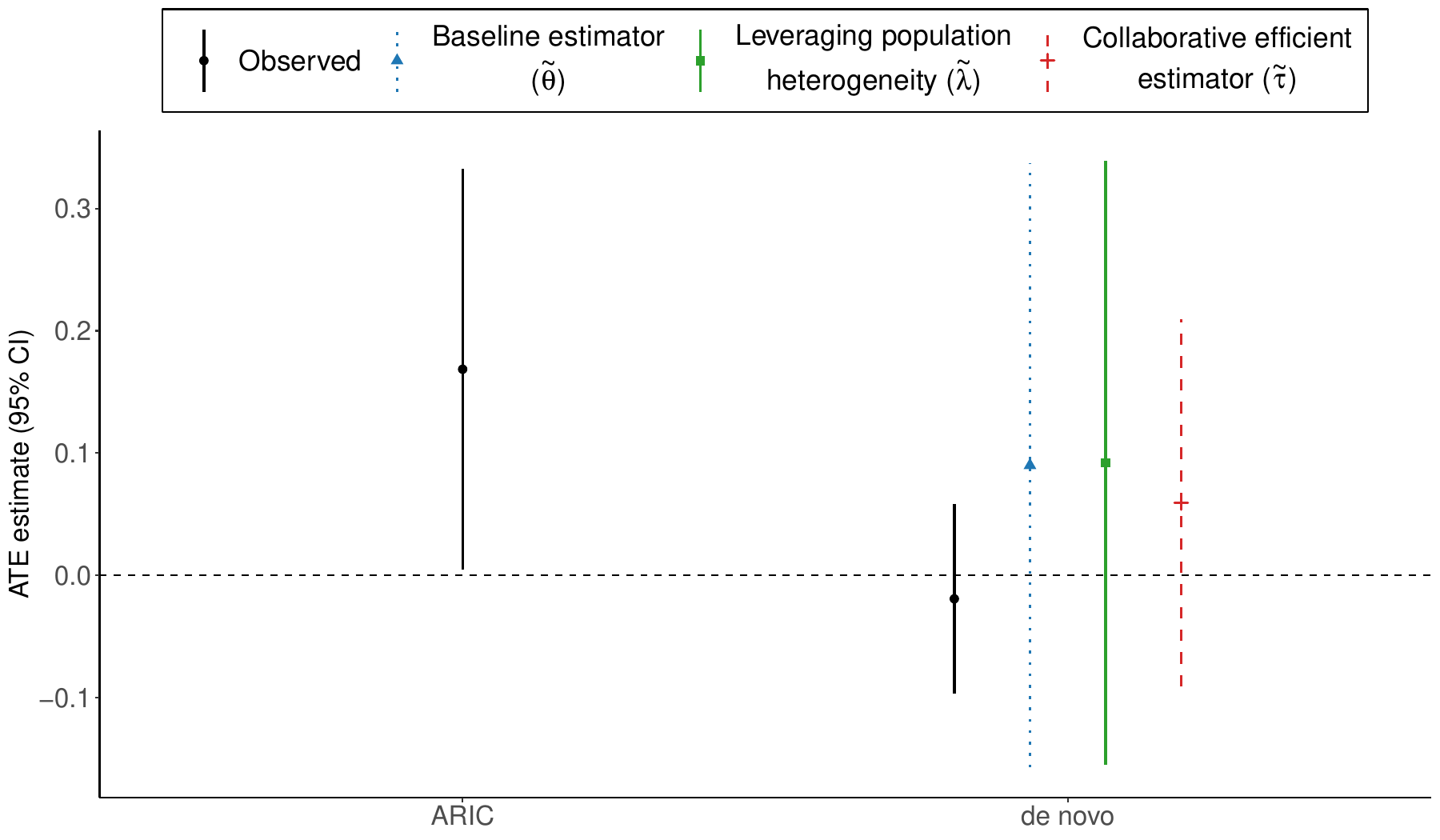}
    \caption{Effect estimates and 95\% confidence intervals using data from the ACHIEVE clinical trial. The ATE is interpreted as the effect of hearing intervention on the 3-year global cognition score compared with health education control. ARIC group is considered as the source population and de novo group is considered as the target population in this example. We show results assuming data generated according to Figure \ref{fig:dag_ic}.}
    \label{fig:achieve_res}
\end{figure}

The effect estimates of receiving a hearing intervention vs. health education on the global cognition score transported from the ARIC to the de novo group are shown in Figure \ref{fig:achieve_res} along with the associated 95\% confidence intervals. Transported point estimates are in the direction of the observed de novo group point estimate, and we are able to obtain consistent the null effects within the de novo group.  Estimates of the $\tilde \lambda$ parameter have similar confidence intervals compared to the existing $\tilde \theta$ parameter; estimates of the $\tilde \tau$ parameter have a narrower confidence interval instead. This result can be explained by simulation for DGM 3, where using the summarized statistics of covariates in $\tilde \tau$ helps with the covariate dimension reduction. The selected subsets of variables in $V$ could be strong predictors of the CATE; thus $\tilde \lambda$ did not show a significant improvement in efficiency over $\tilde \theta $, as reflected in the simulation study and the discussion under Lemma $\ref{lemma:gains}$.

\section{Discussion}
\label{sec:conclusion}
This article develops novel estimators for transported effects in the presence of mediating variables, demonstrating efficiency gains under knowledge of the reasons for study heterogeneity, and introducing a collaborative estimator that can be used when that knowledge is absent. In the presence of multivariate, continuous mediators, our proposed estimators require numerical integration. Proposing methods that are computationally feasible in such cases likely requires further advancements in the estimation of the transported CATE. Another direction of future research is to investigate the lack of regularity of the proposed collaborative estimator and to establish its practical consequences.

\section{Acknowledgments}
Xinyi Zhang, Iván Díaz, and Michele Santacatterina were supported by the National Institute of Allergy and Infectious Diseases of the National Institutes of Health under award number R01AI197146. Kara Rudolph and Iván Díaz were supported by the National Institute on Drug Abuse of the National Institutes of Health under award number R01DA056407. The content is solely the responsibility of the authors and does not necessarily represent the official views of the National Institutes of Health.\\

This research was conducted as a part of the U.S. Census Bureau’s Evidence Building Project Series. Any opinions and conclusions expressed herein are those of the author and do not represent the views of the U.S. Census Bureau. The Census Bureau has ensured appropriate access and use of confidential data and has reviewed these results for disclosure avoidance protection (CBDRB-FY26-0336). \\

We thank members of ACHIEVE Collaborative Research Group for input and suggestions. Members of the ACHIEVE Collaborative Research Group are listed at www.achievestudy.org. The investigators thank the participants and staff of the ACHIEVE and ARIC studies for their important contributions and dedication to the study.\\

This manuscript is the result of funding in whole or in part by the National Institutes of Health (NIH). It is subject to the NIH Public Access Policy. Through acceptance of this federal funding, NIH has been given a right to make this manuscript publicly available in PubMed Central upon the Official Date of Publication, as defined by NIH. The content is solely the responsibility of the authors and does not necessarily represent the official views of the NIH.\\

The ACHIEVE trial was supported by the National Institute on Aging (NIA; R01AG055426), with magnetic resonance imaging of the brain funded by R01AG060502, and with previous pilot study support from the NIA (R34AG046548) and the Eleanor Schwartz Charitable Foundation, in collaboration with the ARIC Study, supported by National Heart, Lung, and Blood Institute (NHLBI) contracts 
(75N92022D00001, 75N92022D00002, 75N92022D00003, 75N92022D00004, 

75N92022D00005). ARIC neurocognitive data were collected using US National Institutes of Health grants (NHLBI, National Institute of Neurological Disorders and Stroke, NIA, and National Institute of Deafness and Other Communication Disorders; U01HL096812, U01HL096814, U01HL096899, U01HL096902, and U01HL096917).\\

Hearing aids, hearing assistive technologies, and related materials used in the ACHIEVE Study were provided at no cost to the researchers or the participants by Sonova/Phonak LLC.\\

Neither the funders of the study nor the sponsoring manufacturer had no role in study design, data collection, data analysis, data interpretation, or writing of the report.

\paragraph{Code availability.}
The code for implementing the proposed estimators is available at
\url{https://github.com/XinyiEmilyZhang/Trans_Mediator_Efficiency}.

\newpage
\bibliographystyle{plainnat}
\bibliography{refs}

\clearpage
\appendix

\begin{bibunit}[plainnat]

\section{Appendix}
\subsection{Identification of $\lambda$}
\label{subsec:identification}

\begin{proof}
We identify our estimand $    \E\left[Y(1) - Y(0)\right | S = 0]$ under non-parametric structural equation models, using the assumptions referenced below. For each $a \in \{0,1\}$,
\begin{align*}
    \E[Y(a) | S = 0] 
    &= \E\left[\E\left[Y(a) | W, S = 0\right] | S = 0 \right]\\
    &= \E\left[\E\left[Y(a) | A = a, W, S = 0\right] | S = 0 \right]&&\text{by \ref{ass:A3}\ref{ass:A3_exch} \& \ref{ass:positivity}\ref{ass:pos_A}} \\
    &= \E\left[\E\left[Y | A = a, W, S = 0\right] | S = 0 \right] &&\text{by consistency}\\
    &=\E\left[\E\left[\E\left[Y | A = a, W, M, Z, S = 0\right] | A = a, W, S = 0 \right]| S = 0 \right] \\
    &= \E\left[\E\left[\E\left[Y | A = a, W, M, Z, S = 1\right] | A = a, W, S = 0 \right]| S = 0 \right] &&\text{by \ref{ass:A3}\ref{ass:A3_S} \& \ref{ass:positivity}\ref{ass:pos_X}} \\
    &=\E\left[h(a,W) | S = 0\right].
\end{align*}

Therefore, 
\begin{align*}
    \lambda = \E\left[Y(1) - Y(0)\right | S = 0] 
    &= \E\left[h(1,W) - h(0,W) | S = 0\right] \\
    &= \E\left[f(W) | S = 0\right] \\
    &= \E\left[\E\left[f(W) |V, S = 0\right] | S = 0\right] \\
    &= \E\left[\E\left[f(W) |V \right] | S = 0\right] &&\text{by \ref{ass:modif}},
\end{align*}
which establishes the identification result.

\end{proof}
\subsection{Identification of $\lambda$ under \ref{ass:A5}} \label{subsec:identificationA5}
\begin{proof}
We can identify $\lambda$ using assumptions in \ref{ass:A5}, which establish identification equivalence as in \ref{ass:A3}. For each $a \in \{0,1\}$,
\begin{align*}
\E[Y(a)\mid S=0] 
&= \E\!\left[\E[Y(a)\mid W,M(a), Z(a), S=0]\mid S=0\right] \\
&= \E\!\left[\E[Y(a)\mid W,M(a), Z(a), S=1]\mid S=0\right]&&\text{by \ref{ass:A5}\ref{ass:A5_S}}\\ 
&= \int \E[Y(a)\mid W=w,M(a)=m, Z(a)=z, S=1]
\\ & \qquad\dd\P(M(a)=m, Z(a)=z\mid W=w, S=0)\dd\P(W=w\mid S=0)\\
&= \int \E[Y\mid A=a, W=w,M=m, Z=z, S=1]
\\ & \qquad \dd\P(M=m, Z=z\mid A=a, W=w, S=0)\dd\P(W=w\mid S=0) \\ & \qquad\qquad \text{by \ref{ass:A5}\ref{ass:A5_exch} \& consistency}
\\
&= \E[h(a,W)\mid S = 0].
\end{align*}

The second half of the proof completes the identification of $\lambda$ following the same steps as Appendix \ref{subsec:identification}.
\end{proof}

\subsection{von Mises expansion for $\lambda$}
\label{subsec:vonmise_lambda}
\begin{proof}\allowdisplaybreaks
\begin{align}
    \hat\p\times\E[&\D_\lambda(O;\hat\eta)] \notag\\ 
    &= \E\bigg[\frac{\{A-\hat e_A(W)\}}{\hat e_A(W)\{1-\hat e_A(W)\}}\{1-\hat q_S(A,W,X)\}\bigg\{\frac{S}{\hat q_S(A,W,X)}-1\bigg\}\{m(A,W,X) - \hat m(A,W,X)\}\bigg]\label{eq:r1}\\
    &+\E\bigg[\frac{\{A-\hat e_A(W)\}}{\hat e_A(W)\{1-\hat e_A(W)\}}\{1-\hat q_S(A,W,X)\}\{m(A,W,X) - \hat m(A,W,X)\}\bigg]\notag\\
    &+\E\bigg[\frac{\{A-\hat e_A(W)\}}{\hat e_A(W)\{1-\hat e_A(W)\}}\{1-S\}\{\hat m(A, W, X) - m(A, W, X)\}\bigg]\notag\\
        &+\E\bigg[\frac{\{A-\hat e_A(W)\}}{\hat e_A(W)\{1-\hat e_A(W)\}}\{1-S\}\{m(A, W, X) - \hat h(A, W)\}\bigg]\notag\\
        &+\E\big[(1-\hat e_S(V))\{\hat f(W) - \hat l(V)\}\big]\notag\\
                &+\E\big[(1-S)\{\hat l(V) - l(V)\}\big]\notag\\
        &=(\ref{eq:r1}) \notag\\
        &+\E\bigg[\frac{\{A-\hat e_A(W)\}}{\hat e_A(W)\{1-\hat e_A(W)\}}\{S-\hat q_S(A,W,X)\}\{m(A,W,X) - \hat m(A,W,X)\}\bigg]\label{eq:r2}\\
                &+\E\bigg[\frac{\{A-\hat e_A(W)\}}{\hat e_A(W)\{1-\hat e_A(W)\}}\{1-S\}\{h(A, W) - \hat h(A, W)\}\bigg]\notag\\
        &+\E\big[(1-\hat e_S(V))\{\hat f(W) - \hat l(V)\}\big]\notag\\
                &+\E\big[(1-e_S(V))\{\hat l(V) - l(V)\}\big]\notag\\
        &=(\ref{eq:r1})+(\ref{eq:r2})+\notag\\
                        &+\E\bigg[\frac{\{A-\hat e_A(W)\}}{\hat e_A(W)\{1-\hat e_A(W)\}}\{1-S\}\{A\times f(W) + g(W) - A\times \hat f(W) - \hat g(W)\}\bigg]\notag\\
                    &+\E\big[(1-\hat e_S(V))\{\hat f(W) - \hat l(V)\}\big]\notag\\
                &+\E\big[(1-e_S(V))\{\hat l(V) - l(V)\}\big]\notag\\
        &=(\ref{eq:r1})+(\ref{eq:r2})+\notag\\
                        &+\E\bigg[\frac{\{A-\hat e_A(W)\}}{\hat e_A(W)\{1-\hat e_A(W)\}}\{1-S\}\{A\times f(W) + g(W) - A\times \hat f(W) - \hat g(W)\}\bigg]\notag\\
                &+\E\big[(\hat e_S(V)-e_S(V))\{\hat l(V) - l(V)\}\big]\label{eq:r21}\\
                                                &+\E\big[(e_S(V)-\hat e_S(V))\{\hat f(W) - f(W)\}\big]\label{eq:r22}\\
                                                &+\E\big[(e_S(V)-\hat e_S(V))\{f(W) - l(V)\}\big]\notag\\                              
        &+\E\big[(1-S)\{\hat f(W) - \E(f(W) \mid V)\}\big]\notag\\
                &=(\ref{eq:r1})+(\ref{eq:r2})+(\ref{eq:r21})+(\ref{eq:r22})\notag\\
                        &+\E\bigg[\frac{A}{\hat e_A(W)}\{1-S\}\{f(W) - \hat f(W) + g(W) - \hat g(W)\}\bigg]\notag\\
                                                &+\E\bigg[-\frac{1-A}{1-\hat e_A(W)}\{1-S\}\{g(W) - \hat g(W)\}\bigg]\notag\\
       &+\E\big[(e_S(V)-\hat e_S(V))\{\E[f(W)\mid V] - l(V)\}\big] \notag\\ 
        &+\E\big[(1-S)\{\hat f(W) - f(W)\}\big]\notag\\
    &=(\ref{eq:r1})+(\ref{eq:r2})+(\ref{eq:r21})+(\ref{eq:r22})\notag\\
    &+\E\bigg[\bigg\{\frac{A}{\hat e_A(W)}-1\bigg\}\{1-S\}\{f(W) - \hat f(W)\}\bigg]\label{eq:r3}\\
    &+ \E\bigg[\bigg\{-\frac{1-A}{1-\hat e_A(W)}+\frac{A}{\hat e_A(W)}\bigg\}\{1-S\}\{g(W) - \hat g(W)\}\bigg].\label{eq:r4}
\end{align}

Thus we have $R_\D(\eta,\hat{\eta}) = (\ref{eq:r1})+(\ref{eq:r2})+(\ref{eq:r21})+(\ref{eq:r22})+ (\ref{eq:r3}) + (\ref{eq:r4})$.

Assume there exists $\delta > 0$ such that $\delta<\hat e_A(W)<1-\delta$ and $\delta<\hat q_S(A,W,X)<1-\delta$. We have
    \[\P\bigg[\bigg|\frac{A-\hat e_A(W)}{\hat e_A(W)\{1-\hat e_A(W)\}}\bigg|\leq \epsilon^{-1} \bigg]=1\]
for some $\epsilon > 0$. Conditioning on $(A,W,X)$ in (\ref{eq:r1}) replaces $S$ by $q_S(A,W,X)$, so that
(\ref{eq:r1}) is bounded in absolute value by
\[
\epsilon^{-1}\E\bigg[\bigg|\frac{q_S(A,W,X)}{\hat q_S(A,W,X)}-1\bigg|\,\big|m(A,W,X) - \hat m(A,W,X)\big|\bigg]
\]




Similarly, conditioning on $(A,W,X)$ in (\ref{eq:r2}) replaces $S$ by $q_S(A,W,X)$, so that
(\ref{eq:r2}) is bounded in absolute value by
\[
\epsilon^{-1}\E\Big[\big|q_S(A,W,X) - \hat{q}_S(A,W,X)\big|\,\big|m(A,W,X) - \hat m(A,W,X)\big|\Big]
\]



By similar arguments, (\ref{eq:r3}) and (\ref{eq:r4}) are bounded by

\[\epsilon^{-1}\E\bigg[\bigg|\frac{e_A(W)}{\hat e_A(W)}-1\bigg|\,\big|f(W) - \hat f(W)\big|\bigg],\]
and
\[\epsilon^{-1}\E\bigg[\bigg|-\frac{1-e_A(W)}{1-\hat e_A(W)}+\frac{e_A(W)}{\hat e_A(W)}\bigg|\,\big|g(W) - \hat g(W)\big|\bigg].\]



The theorem follows after noticing that
\begin{align}
\E[\D_{\lambda}(O;\hat\eta) - \D_{\hat\lambda}(O;\hat\eta)] &=(\hat\lambda-\lambda)\notag\\
&+\left\{\frac{\p}{\hat\p}-1\right\}(\hat\lambda-\lambda)\label{eq:rlambda}.
\end{align}
\end{proof}

\subsection{Efficiency gains for $\lambda$}
\label{subsec:eff_gain}
\begin{proof} Since $\D_\lambda(O;\eta)$ and $\D_\theta(O;\xi)$ differ only in their final two terms, it suffices to compare the variances of those differing components.

\begin{align*}
\P(S = 0)^2\Big[\var\left\{\D_\lambda(O;\eta)\right\} - \var\left\{\D_\theta(O;\xi)\right\}\Big] 
&= \var\Big[\P(S=0 \mid V)\{f(W) - \E(f(W) \mid V)\}\Big] \\
&\quad + \var\Big[\one\{S=0\} \{\E(f(W) \mid V) - \lambda\}\Big] \\
&\quad - \var\Big[\one\{S=0\} (f(W) - \theta)\Big].
\end{align*}

For the first term, let
\[
R := \P(S=0 \mid V)\{f(W) - \E\{f(W)\mid V\}\}.
\]
Then
\[
\E(R)
= \E\left[\E\big\{R\mid V\big\}\right]
= \E\left[\P(S=0\mid V)\E\big\{f(W)-\E\{f(W)\mid V\}\mid V\big\}\right]
=0.
\]
Therefore,
\[
\var(R)=\E(R^2)
= \E\Big[\P(S=0\mid V)^2\{f(W)-\E\{f(W)\mid V\}\}^2\Big].
\]

Next, since $\theta=\E\{f(W)\mid S=0\}$,
\[
\E[\one\{S=0\}(f(W)-\theta)] = 0,
\]
and so
\begin{align*}
\var\Big[\one\{S=0\}(f(W)-\theta)\Big]
&= \E\Big[\one\{S=0\}(f(W)-\theta)^2\Big] \\
&= \E\left[\E\Big\{\one\{S=0\}(f(W)-\theta)^2\mid V\Big\}\right] \\
&=\E\Big[\P(S=0\mid V)\E\{(f(W)-\theta)^2\mid V\}\Big],
\end{align*}

where the last equality follows from $S\indep W\mid V$ under Assumption~\ref{ass:modif}.

Similarly, 
\[
\var\Big[\one\{S=0\}\{\E\{f(W)\mid V\}-\lambda\}\Big]
=
\E\Big[\P(S=0\mid V)\{\E\{f(W)\mid V\}-\theta\}^2\Big].
\]

Now use the conditional bias-variance decomposition
\[
\E\{(f(W)-\theta)^2\mid V\}
=
\var\{f(W)\mid V\} + \{\E\{f(W)\mid V\}-\theta\}^2.
\]
Therefore,
\begin{align*}
&\var\Big[\one\{S=0\}\big\{\E\{f(W)\mid V\}-\lambda\big\}\Big]
- \var\Big[\one\{S=0\}\big\{f(W)-\theta\big\}\Big] \\
&= - \E\Big[\P(S=0\mid V)\var\{f(W)\mid V\}\Big] \\
&= - \E\Big[\P(S=0\mid V)\E\big[\big\{f(W)-\E\{f(W)\mid V\}\big\}^2 \mid V\big]\Big] \\
&=  - \E\Big\{\P(S=0\mid V)[f(W)-\E\{f(W)\mid V\}]^2\Big\}.
\end{align*}

Combining terms,
\begin{align*}
&\P(S = 0)^2\Big[\var\{\D_\lambda(O;\eta)\} - \var\{\D_\theta(O;\xi)\}\Big] \\
&= \E\Big\{\P(S=0\mid V)^2[f(W)-\E\{f(W)\mid V\}]^2\Big\}
 - \E\Big\{\P(S=0\mid V)[f(W)-\E\{f(W)\mid V\}]^2\Big\} \\
&= \E\Big\{\P(S=0\mid V)\{\P(S=0\mid V)-1\}[f(W)-\E\{f(W)\mid V\}]^2\Big\} \le 0.
\end{align*}

\end{proof}

\subsection{von Mises expansion for $\tau$}
\label{subsec:vonmise_tau}
\begin{proof}\allowdisplaybreaks
\begin{align}
    \hat\p\times\E[&\U_\tau(O;\hat\gamma)] \notag\\ 
    &= \E\bigg[\frac{\{A-\hat e_A(W)\}}{\hat e_A(W)\{1-\hat e_A(W)\}}\{1-\hat r_S(A,W,X)\}\bigg\{\frac{S}{\hat r_S(A,W,X)}-1\bigg\}\{m(A,W,X) - \hat m(A,W,X)\}\bigg]\label{eq:r10}\\
    &+\E\bigg[\frac{\{A-\hat e_A(W)\}}{\hat e_A(W)\{1-\hat e_A(W)\}}\{1-\hat r_S(A,W,X)\}\{m(A,W,X) - \hat m(A,W,X)\}\bigg]\notag\\
    &+\E\bigg[\frac{\{A-\hat e_A(W)\}}{\hat e_A(W)\{1-\hat e_A(W)\}}\{1-S\}\{\hat m(A, W, X) - m(A, W, X)\}\bigg]\notag\\
        &+\E\bigg[\frac{\{A-\hat e_A(W)\}}{\hat e_A(W)\{1-\hat e_A(W)\}}\{1-S\}\{m(A, W, X) - \hat h(A, W)\}\bigg]\notag\\
        &+\E\big[(1-\hat e_S(W))\{\hat f(W) - \hat k(W)\}\big]\notag\\
                &+\E\big[(1-S)\{\hat k(W) - k(W)\}\big]\notag\\
        &=(\ref{eq:r10}) \notag\\
        &+\E\bigg[\frac{\{A-\hat e_A(W)\}}{\hat e_A(W)\{1-\hat e_A(W)\}}\{S-\hat r_S(A,W,X)\}\{m(A,W,X) - \hat m(A,W,X)\}\bigg]\label{eq:r11}\\
                &+\E\bigg[\frac{\{A-\hat e_A(W)\}}{\hat e_A(W)\{1-\hat e_A(W)\}}\{1-S\}\{h(A, W) - \hat h(A, W)\}\bigg]\notag\\
        &+\E\big[(1-\hat e_S(W))\{\hat f(W) - \hat k(W)\}\big]\notag\\
                &+\E\big[(1-e_S(W))\{\hat k(W) - k(W)\}\big]\notag\\
        &=(\ref{eq:r10})+(\ref{eq:r11})+\notag\\
                        &+\E\bigg[\frac{\{A-\hat e_A(W)\}}{\hat e_A(W)\{1-\hat e_A(W)\}}\{1-S\}\{A\times f(W) + g(W) - A\times \hat f(W) - \hat g(W)\}\bigg]\notag\\
                &+\E\big[(\hat e_S(W)-e_S(W))\{\hat k(W) - k(W)\}\big]\label{eq:r12}\\
                                                &+\E\big[(e_S(W)-\hat e_S(W))\{\hat f(W) - f(W)\}\big]\label{eq:r13}\\
                                                &+\E\big[(e_S(W)-\hat e_S(W))\{f(W) - k(W)\}\big]\notag\\                              
        &+\E\big[(1-S)\{\hat f(W) - f(W)\}\big]\notag\\
                &=(\ref{eq:r10})+(\ref{eq:r11})+(\ref{eq:r12})+(\ref{eq:r13})\notag\\
                        &+\E\bigg[\frac{A}{\hat e_A(W)}\{1-S\}\{f(W) - \hat f(W) + g(W) - \hat g(W)\}\bigg]\notag\\
                                                &+\E\bigg[-\frac{1-A}{1-\hat e_A(W)}\{1-S\}\{g(W) - \hat g(W)\}\bigg]\notag\\
       &+\E\big[(e_S(W)-\hat e_S(W))\{\E[f(W)\mid e_S(W),\hat e_S(W)] - k(W)\}\big]\label{eq:r14}\\ 
        &+\E\big[(1-S)\{\hat f(W) - f(W)\}\big]\notag\\
    &=(\ref{eq:r10})+(\ref{eq:r11})+(\ref{eq:r12})+(\ref{eq:r13})+(\ref{eq:r14})\notag\\
    &+\E\bigg[\E\bigg[\bigg\{\frac{A}{\hat e_A(W)}-1\bigg\}\{1-S\}\{f(W) - \hat f(W)\}\bigg | W\bigg]\bigg] \notag \\
    &+ \E\bigg[\E\bigg[\bigg\{-\frac{1-A}{1-\hat e_A(W)}+\frac{A}{\hat e_A(W)}\bigg\}\{1-S\}\{g(W) - \hat g(W)\}\bigg | W\bigg]\bigg]. \notag\\
    &= (\ref{eq:r10})+(\ref{eq:r11})+(\ref{eq:r12})+(\ref{eq:r13})+(\ref{eq:r14})\notag\\
    &+\E\bigg[\bigg\{\frac{e_A(W)}{\hat e_A(W)}-1\bigg\}\{1-S\}\{f(W) - \hat f(W)\}\bigg]\label{eq:r15}\\
    &+ \E\bigg[\bigg\{-\frac{1-e_A(W)}{1-\hat e_A(W)}+\frac{e_A(W)}{\hat e_A(W)}\bigg\}\{1-S\}\{g(W) - \hat g(W)\}\bigg].\label{eq:r16}
\end{align}

Next, define $R_\U(\gamma,\hat\gamma)=(\ref{eq:r10})+(\ref{eq:r11})+(\ref{eq:r12})+(\ref{eq:r13})+(\ref{eq:r14})+(\ref{eq:r15})+(\ref{eq:r16})$. By assumption ($\delta<\hat e_A(W)<1-\delta$ for $\delta>0$), we have 
\[\P\bigg[\bigg|\frac{A-\hat e_A(W)}{\hat e_A(W)\{1-\hat e_A(W)\}}\bigg|\leq \epsilon^{-1} \bigg]=1\]
for some $\epsilon$ so that (\ref{eq:r10}) is bounded above by 
  \[
  \epsilon^{-1}\E\bigg[\bigg|\bigg\{\frac{r_S^\star(A,W,X)}{\hat r_S(A,W,X)}-1\bigg\}\bigg| \bigg|\bigg\{m(A,W,X) - \hat m(A,W,X)\} \bigg|\bigg]
  \]
where we denoted $r_S^\star(a,w,x)=\P(S=1\mid \hat m(A,W,X) = \hat m(a,w,x), m(A,W,X) = m(a,w,x))$.
Following similar arguments, (\ref{eq:r11}) is bounded above by 

\[
\epsilon^{-1}\E\big[\big|\{r_S^\star(A,W,X)-\hat r_S(A,W,X)\}\big|\big|\{m(A,W,X) - \hat m(A,W,X)\}\big|\big].
\]

Similarly, (\ref{eq:r15}) and (\ref{eq:r16}) are bounded in absolute value by
\[\epsilon^{-1}\E\bigg[\bigg|\frac{e_A(W)}{\hat e_A(W)}-1\bigg|\,\big|f(W) - \hat f(W)\big|\bigg],\]
and
\[\epsilon^{-1}\E\bigg[\bigg|-\frac{1-e_A(W)}{1-\hat e_A(W)}+\frac{e_A(W)}{\hat e_A(W)}\bigg|\,\big|g(W) - \hat g(W)\big|\bigg].\]


The theorem follows after noticing that 
\begin{align}
\E[\U_{\tau}(O;\hat\gamma) - \U_{\hat\tau}(O;\hat\gamma)]&=\hat\tau-\tau\notag\\
&+\left\{\frac{\p}{\hat\p}-1\right\}(\hat\tau-\tau)\label{eq:rtau}.
\end{align}
\end{proof}

\subsection{Doubly robust transformation of $f(W)$}
\label{subsec:DR_CATE}

\begin{proof}
For $a \in \{0,1\}$, we define $h_1 := \E[h(1,w)]$ and $h_0 := \E[h(0,w)] $. Without loss of generality, we compute the influence function of $h_1$. We begin by deriving the influence function for $h(1,w)$:

\begin{align*}
& \mathbb{IF}\{h(1,w)\} \\ &= \sum_x \Bigg\{ \frac{\one(A = 1, S = 1, W = w, X = x)}{\P(A = 1, S = 1, W = w, X = x)} \Big[ Y - \E(Y | A = 1, S = 1, W = w, X = x) \Big] \\
&\qquad \times \P(X = x | A = 1, W = w, S = 0) \\
&\qquad + \E(Y | A = 1, W = w, X = x, S = 1) \times \frac{\one(A = 1, W = w, S = 0)}{\P(A = 1, W = w, S = 0)} \\
&\qquad \times \Big[ \one(X = x) - \P(X = x | A = 1, W = w, S = 0) \Big] \Bigg\} \\
&= \sum_x \Bigg\{ \frac{\one(A = 1, S = 1, W = w, X = x)}{\P(A = 1, S = 1, W = w)} \frac{\P(S = 0 | X = x, A = 1, W = w)}{\P(S = 1 | X = x, A = 1, W = w)} \frac{\P(S = 1, A = 1, W = w)}{\P(S = 0, A = 1, W = w)} \\
&\qquad \times \Big[ Y - \E(Y | A = 1, S = 1, W = w, X = x) \Big] \\
&\qquad + \E(Y | A = 1, W = w, X = x, S = 1) \\
&\qquad \times \frac{\one(A = 1, W = w, S = 0)}{\P(A = 1, W = w, S = 0)} \Big[ \one(X = x) - \P(X = x | A = 1, W = w, S = 0) \Big] \Bigg\} \\
&= \frac{\one(A = 1, S = 1, W = w)}{\P(A = 1, W = w, S = 0)} \frac{\P(S = 0 | X, A = 1, W = w)}{\P(S = 1 | X, A = 1, W = w)} \Big[ Y - \E(Y | A = 1, S = 1, W = w, X) \Big] \\
&\qquad + \frac{\one(A = 1, S = 0, W = w)}{\P(A = 1, S = 0, W = w)} \Big[ \E(Y | A = 1, S = 1, W = w, X) - h(1,w) \Big],
\end{align*}
and
\begin{align*}
  &\mathbb{IF}\{h_1\}\\
  &= \mathbb{IF}\{\E[h(1,w)]\} = \mathbb{IF}\left\{\sum_w h(1,w)p(w)\right\} = \sum_w \mathbb{IF}\{h(1,w)p(w)\} \\
  &= \sum_w \left\{ \mathbb{IF}\{h(1,w)\} p(w) + \mathbb{IF}\{p(w)\} h(1,w) \right\} \\
  &= \sum_w \Bigg\{\frac{\one(A = 1, S = 1, W = w) \times p(w)}{\P(A = 1, S = 0, W = w)}\frac{\P(S = 0 \mid X, A = 1, W=w)}{\P(S = 1 \mid X, A = 1, W = w)} \\
&\qquad \times [Y - \E(Y \mid A = 1, S = 1, W=w, X)] \\
  &\qquad + \frac{\one(A = 1, W = w, S = 0)\times p(w)}{\P(A = 1, W = w, S = 0)} \left[\E(Y \mid A = 1, S = 1, W = w, X) - h(1,w)\right]\\
  &\qquad + h(1,w)[\one(W= w) - p(w)]\Bigg\} \\
  &= \frac{\one(A = 1, S = 1)}{\P(A = 1 | S = 0, W) \P(S = 0 | W)} \times \frac{\P(S = 0 \mid X, A = 1, W)}{\P(S = 1 | X, A = 1, W)}[Y - \E(Y | A = 1, S = 1, W, X)] \\
  &\qquad + \frac{\one(A = 1, S=0)}{\P(A = 1 | S = 0, W) \P(S = 0 | W)} \times [\E(Y | A = 1, S = 1, W, X) - h(1,W)] \\
  &\qquad + h(1,W) - h_1.
\end{align*}

Given nuisances parameter $\kappa = (e_A,  e_S, q_S, m, p_X, h)$, we define $  \D_{h,1}(O; \kappa)$ as
\begin{align*}
    &\D_{h,1}(O; \kappa)\\ & = \frac{\one(A = 1, S = 1)}{\P(A = 1 | S = 0, W) \P(S = 0 | W)} \times \frac{\P(S = 0 \mid X, A = 1, W)}{\P(S = 1 | X, A = 1, W)}[Y - \E(Y | A = 1, S = 1, W, X)] \\
  &\qquad + \frac{\one(A = 1, S=0)}{\P(A = 1 | S = 0, W) \P(S = 0 | W)} \times [\E(Y | A = 1, S = 1, W, X) - h(1,W)] \\
  &\qquad + h(1,W) \\
  &=\frac{\one(A = 1, S = 1)}{e_A(W) \{1 - e_S(W)\}} \times \frac{1-q_S(1,W,X)}{q_S(1,W,X)}[Y - m(1,W,X)] \\
  &\qquad + \frac{\one(A = 1, S=0)}{e_A(W) \{1 - e_S(W)\}} \times [m(1,W,X) - h(1,W)] \\
  &\qquad + h(1,W). 
\end{align*}

Now we prove $\E[\D_{h,1}(O; \hat{\kappa}) \mid W  = w] = h(1,w) + R(\hat{\kappa}, \kappa)$, where $R(\hat{\kappa}, \kappa)$ is of the second order.
{\footnotesize
\begin{align}
  &\E \Big[\D_{h,1}(O; \hat{\kappa})  - h(1,W)  \big| W \Big] \notag \\
   &= \E\Bigg\{\frac{1 - \hat{q}_S(1,W,X)}{\hat{q}_S(1,W,X)} \frac{A}{\hat{e}_A(W)}\frac{S}{1 - \hat{e}_S(W)}\Big[m(1, W, X) - \hat{m}(1, W, X)\Big] \notag\\
   & \qquad + \frac{e_A(W)}{\hat{e}_A(W)} \frac{1-e_S(W)}{1 - \hat{e}_S(W)}\Big[\sum_x\hat{m}(1, W, X) [p_X(1,W) -  \hat{p}_X(1,W)]\Big] \notag\\
   & \qquad + \sum_x \hat{m}(1, W, X) \hat{p}_X(1,W) - \sum_x m(1, W, X) p_X(1,W) \bigg| W \Bigg\} \notag\\
   &= \E\Bigg\{\Big[\frac{1 - \hat{q}_S(1,W,X)}{\hat{q}_S(1,W,X)\hat{e}_A(W)(1 - \hat{e}_S(W))} - \frac{1 - q_S(1,W,X)}{q_S(1,W,X)e_A(W)(1 - e_S(W))}\Big] AS \Big[m(1, W, X) - \hat{m}(1, W, X) \Big] \bigg| W \Bigg\}\label{eq:r17} \\
   &\qquad + \E\Bigg\{\frac{1 - q_S(1,W,X)}{q_S(1,W,X) e_A(W) (1 - e_S(W))} A S [m(1, W, X) - \hat{m}(1, W, X) ] \bigg| W  \Bigg\} \notag  \\
   & \qquad  + \E\Bigg\{ \frac{e_A(W)}{\hat{e}_A(W)} \frac{1 - e_S(W)}{1-\hat{e}_S(W)} \Big[\sum_x \hat{m}(1, W, X) [p_X(1,W) - \hat{p}_X(1,W)]\Big] \bigg| W  \Bigg\} \notag \\
   & \qquad + \E\Bigg\{ \sum_x \hat{m}(1, W, X) \hat{p}_X(1,W) - \sum_x m(1, W, X) p_X(1,W) \bigg| W \Bigg\} \notag \\
   &= (\ref{eq:r17}) \notag \\
   & \qquad +  \sum_x \P(X=x \mid A=1, W, S=1)\,\frac{\P(S = 0 \mid X, A = 1, W)}{\P(S = 1 \mid X, A = 1 ,W)} \frac{\P(A = 1 \mid W, S = 1)}{\P(A = 1 \mid W, S = 0)} \frac{\P(S = 1 \mid W)}{\P(S = 0 \mid W)}  [m(1, W, X) - \hat{m}(1, W, X)]  \notag\\
   & \qquad  + \E\Bigg\{ \Big[\frac{e_A(W)}{\hat{e}_A(W)} \frac{1 - e_S(W)}{1-\hat{e}_S(W)}  - 1\Big] \Big[\sum_x \hat{m}(1, W, X) [p_X(1,W) - \hat{p}_X(1,W)]\Big] \bigg| W \Bigg\} \label{eq:r18}  \\ 
   & \qquad + \E\Bigg\{\sum_x \hat{m}(1, W, X) [p_X(1,W) - \hat{p}_X(1,W)] \bigg| W \Bigg\}  \notag \\
    & \qquad + \E\Bigg\{ \sum_x \hat{m}(1, W, X) \hat{p}_X(1,W) - \sum_x m(1, W, X) p_X(1,W) \bigg| W \Bigg\} \notag \\
    &= (\ref{eq:r17}) + (\ref{eq:r18} ) \notag \\
    & \qquad + \sum_x p_X [m(1, W, X) - \hat{m}(1, W, X)] +  \E\Bigg\{\sum_x \hat{m}(1, W, X) [p_X(1,W) - \hat{p}_X(1,W)] \bigg| W \Bigg\} \notag \\
    & \qquad + \E\Bigg\{ \sum_x \hat{m}(1, W, X) \hat{p}_X(1,W) - \sum_x m(1, W, X) p_X(1,W) \bigg| W \Bigg\} \notag \\
    &= (\ref{eq:r17}) + (\ref{eq:r18} ) \notag.
\end{align}
}

The influence function of $h_0 = \E[h(0,w)]$ follows similar steps: 
\begin{align*}
     &\mathbb{IF}\{h_0\}\\   &= \frac{\one(A = 0, S = 1)}{\P(A = 0 | S = 0, W) \P(S = 0 | W)} \times \frac{\P(S = 0 \mid X, A = 0, W)}{\P(S = 1 | X, A = 0, W)}[Y - \E(Y | A = 0, S = 1, W, X)] \\
  &\quad + \frac{\one(A = 0, S=0)}{\P(A = 0 | S = 0, W) \P(S = 0 | W)} \times [\E(Y | A = 0, S = 1, W, X) - h(0,W)] \\
  &\quad + h(0,W) - h_0\\
  &= \frac{\one(A = 0, S = 1)}{\{1 - e_A(W)\} \{1 - e_S(W)\}} \times \frac{1-q_S(0,W,X)}{q_S(0,W,X)}[Y - m(0,W,X)] \\
  &\qquad + \frac{\one(A = 0, S=0)}{\{1 - e_A(W)\} \{1 - e_S(W)\}} \times [m(0,W,X) - h(0,W)] \\
  &\qquad + h(0,W) - h_0 \\
  &= \D_{h,0}(O; \kappa) - h_0.
\end{align*}

The proof for
$\E[\D_{h,0}(O; \hat{\kappa}) \mid W  = w] = h(0,w) + R(\hat{\kappa}, \kappa)$, where $R(\hat{\kappa}, \kappa)$ is of the second order follows a similar manner. Combining, we obtain the doubly robust transformed CATE---for arbitrary $\kappa'$ we have $f(w) = \E[\D_{h}(O; \kappa') \mid W=w]$.

\end{proof}

\subsection{Covariate descriptions for ACHIEVE application analysis}
\label{subsec:ACHIEVE variables}

\textbf{Demographic.} Each participant self-reported their date of birth, sex, race, education, former occupation, and annual income. Date of birth was used to calculate age. The field site that recruited each participant was documented by staff.\\

\textbf{Genetic.} The Human Genetics Center at the University of Texas, Houston analyzed DNA extracted from participant blood samples \citep{blair2005apoe}. The TaqMan assay (Applied Biosystems, Foster City, CA) detected APOE variants at codons 130 and 176 and determined the presence of 0, 1, or 2 $\varepsilon4$ alleles. \\

\textbf{Hearing.} Objective hearing was quantified through audiometry performed in sound-attenuating rooms. Pure-tone air- and bone-conduction thresholds were assessed in each ear using a modified Hughson-Westlake \citep{hughson1944manual} psychophysical bracketing method \citep{carhart1959preferred}. Pure tone average was defined as the mean in the better-hearing ear across the frequencies 0.5, 1, 2, and 4 kHz. Communicative function was measured by the 10-item screening version of the Hearing Handicap Inventory for the Elderly \citep{tomioka2013hearing, ventry1983identification}. Loud noise exposure was quantified based on self-reported lifetime exposure to firearms, job-related loud noise for more than 10 hours per week, or very loud noise for more than 10 hours per week outside of a job \citep{kim2020association}. Speech in noise ability was measured by the Quick Speech in Noise (QuickSIN) test \citep{killion2004development}.\\

\textbf{Anthropometric.} Body weight was measured to the nearest 0.1 kilogram, and height was measured to the nearest centimeter. Body mass index (BMI) was calculated as weight in kilograms divided by height in meters squared. Waist circumference was measured to the nearest centimeter using the smallest circumference between the lower ribs and iliac crests and hip circumference was measured using the greatest circumference between the iliac crest and thighs. The ratio of the waist-to-hip circumference was computed. \\

\textbf{Cardiovascular.} Systolic blood pressure (SBP) and diastolic blood pressure (DBP) were measured using the Omron HEM-907 XL oscillometric automated sphygmomanometer (Omron Healthcare, Kyoto, Japan). Resting heart rate was calculated from a 2-minute supine 12‐lead electrocardiogram recording using standardized methods \citep{liao1996computer}.\\

\textbf{Self-reported Medical Conditions.} Hypertension was defined as SBP $\geq 140$ mmHg, DBP $\geq 90$ mmHg, use of anti-hypertensive medication, or self-reported physician diagnosis. Diabetes, coronary heart disease, myocardial infarction, stroke, and Parkinson’s disease were defined based on medication use or self-reported physician diagnosis. Traumatic brain injury and seizures were determined by self-reported physician diagnosis.
Lifestyle. Current, former, or never use of cigarettes or alcohol were ascertained by self-report. Leisure-time and sport-related physical activity were measured by the Baecke questionnaire \citep{baecke1982short}. \\

\textbf{Physical Function.} Lower extremity function was quantified from balance tests (tandem, semi-tandem, side-by-side), repeated chair stands, and a 4-meter walk. Each test was rated on a scale from 0 to 4 \citep{guralnik1994short}. Grip strength in kilograms of force in the dominant hand was measured using a Jamar Hydraulic Hand Dynamometer (Chicago, IL).\\

\textbf{Functional Status.} Functional limitations \citep{nagi1976epidemiology, rosow1966guttman} were measured by 5 items with a 4-point scale that determined whether the participant had difficulty walking a quarter mile, walking up 10 steps, bending their body (stooping, crouching, or kneeling), lifting or carrying, or standing up. Participants also self-reported whether they needed a walking aid, a special eating utensil, or devices to help dress themselves. Difficulty with instrumental activities of daily living \citep{lawton1969assessment} was measured by 3 items with a 4-point scale that ascertained whether the participant was able to do chores, prepare meals, or manage money on their own. A fourth item asked if the participant needed help with household chores, doing necessary business, shopping or getting around. Activities of daily living \citep{katz1963studies} were assessed by 4 items with a 4-point scale that asked about the participant’s ability to walk between rooms, get out of bed, eat, and dress themselves. Participants also self-reported whether they needed help with personal care. \\

\textbf{Mental Health.} Depressive symptomatology was measured using the 11-item Center for Epidemiologic Studies Depression scale \citep{radloff1977ces, kohout1993two} validated for older adults \citep{gellis2010assessment}. \\

\textbf{Cognition.} Cognition was assessed by the Mini Mental State Exam \citep{folstein1975mini} and a 10-test cognitive battery. The battery included the Digit Span Backwards \citep{wechsler1987wechsler}, Boston Naming Test \citep{williams1989boston}, Word Fluency Test \citep{benton1976multilingual}, Animal Naming Score \citep{benton1976multilingual}, Digit Symbol Substitution \citep{wechsler1987wechsler}, Trail Making Tests A and B \citep{reitan1958validity}, Incidental Learning \citep{ryan2001wais}, Logical Memory Test \citep{wechsler1987wechsler}, and the Delayed Word Recall \citep{knopman1989verbal}. Scores from the 10-test cognitive battery were used to compute a factor score of global cognition \citep{gross2015application}. Select tests were used to compute separate factor scores for predefined cognitive domains \citep{gross2015application, rawlings2016factor} of executive function (Digit Symbol Substitution, Trail Making Tests A and B), language (Boston Naming Test, Word Fluency Test, and Animal Naming Score), and memory (Incidental Learning, Logical Memory Test, Delayed Word Recall).

\putbib[refs]
\end{bibunit}

\end{document}